\documentclass[numberwithinsect,a4paper,UKenglish,cleveref, autoref, thm-restate]{lipics-v2021}

\nolinenumbers

\title{On finding short reconfiguration sequences between independent sets} 

\titlerunning{On finding short reconfiguration sequences between independent sets} 

\author{Akanksha Agrawal}{Indian Institute of Technology Madras, Chennai, India}{akanksha@cse.iitm.ac.in}{}{}

\author{Soumita Hait}{Indian Institute of Technology, Kharagpur, India}{soumitahait7321@gmail.com}{}{}

\author{Amer~E.~Mouawad}{American University of Beirut, Lebanon \and University of Bremen, Germany}{aa368@aub.edu.lb}{https://orcid.org/0000-0003-2481-4968}{Research supported by
the Alexander von Humboldt Foundation, by PHC Cedre project 2022 ``PLR'', and partially supported by URB project ``A theory of change through the lens of reconfiguration''.}

\authorrunning{A. Agrawal, S. Hait, and A.~E.~Mouawad} 

\Copyright{Akanksha Agrawal, Soumita Hait, and Amer~E.~Mouawad} 



\category{} 

\relatedversion{} 

\hideLIPIcs  

\EventEditors{Petr A. Golovach and Meirav Zehavi}
\EventNoEds{2}
\EventLongTitle{16th International Symposium on Parameterized and Exact Computation (IPEC 2021)}
\EventShortTitle{IPEC 2021}
\EventAcronym{IPEC}
\EventYear{2021}
\EventDate{September 8--10, 2021}
\EventLocation{Lisbon, Portugal}
\EventLogo{}
\SeriesVolume{214}
\ArticleNo{2}

\newcommand{\C}[1]{\mathcal{#1}}

\usepackage{todonotes}

\newcommand{\WOH}{\textsf{W[1]}-hard}

\newcommand{\mc}{\mathcal}
\newcommand{\Oh}{\mathcal{O}}

\begin{document}

\maketitle

\begin{abstract}
Assume we are given a graph $G$, two independent sets $S$ and $T$ in $G$ of size $k \geq 1$, and a positive integer $\ell \geq 1$. 
The goal is to decide whether there exists a sequence $\langle I_0, I_1, ..., I_\ell \rangle$ of independent sets such that for all $j \in \{0,\ldots,\ell-1\}$ the 
set $I_j$ is an independent set of size~$k$, $I_0 = S$, $I_\ell = T$, and $I_{j+1}$ is obtained from $I_j$ by a predetermined reconfiguration rule. 
We consider two reconfiguration rules, namely token sliding and token jumping. Intuitively, we view each independent set as a collection of tokens placed 
on the vertices of the graph. Then, the \textsc{Token Sliding Optimization (TSO)} problem asks whether there exists a sequence of at most $\ell$ steps 
that transforms $S$ into $T$, where at each step we are allowed to slide one token from a vertex to an unoccupied neighboring vertex (while maintaining independence). 
In the \textsc{Token Jumping Optimization (TJO)} problem, at each step, we are allowed to jump one token from a vertex to any other unoccupied vertex of 
the graph (as long as we maintain independence). Both \textsc{TSO} and \textsc{TJO} are known to be fixed-parameter tractable when parameterized by $\ell$ 
on nowhere dense classes of graphs. In this work, we investigate the boundary of tractability for sparse classes of graphs. We show that both problems are 
fixed-parameter tractable for parameter $k + \ell + d$ on $d$-degenerate graphs as well as for parameter $|M| + \ell + \Delta$ on graphs having a modulator 
$M$ whose deletion leaves a graph of maximum degree $\Delta$. We complement these result by showing that for parameter $\ell$ alone both problems 
become W[1]-hard already on $2$-degenerate graphs. Our positive result makes use of the notion of independence covering families introduced by 
Lokshtanov et al.~\cite{DBLP:journals/talg/LokshtanovPSSZ20}. Finally, we show as a side result that using such families one can obtain a simpler 
and unified algorithm for the standard \textsc{Token Jumping Reachability} problem (a.k.a. \textsc{Token Jumping}) parameterized by $k$ on both 
degenerate and nowhere dense classes of graphs. 
\end{abstract}

\section{Introduction}\label{sec:intro}
Given a simple undirected graph $G$, a set of 
vertices $I \subseteq V(G)$ is an \emph{independent set} if the vertices of $I$ are pairwise non-adjacent. 
Finding an independent set of size $k$, i.e., the \textsc{Independent Set (IS)} problem, is known to be NP-complete~\cite{DBLP:conf/coco/Karp72} and 
W[1]-complete parameterized by solution size $k$~\cite{DBLP:journals/siamcomp/DowneyF95}. 
We view an independent set as a collection of $k$ tokens placed on the vertices of a graph such that no two tokens are placed on adjacent vertices. 
This gives rise to two natural adjacency relations between independent sets (or token configurations), also called \emph{reconfiguration steps}. 
These reconfiguration steps, in turn, give rise to several \emph{combinatorial reconfiguration} 
problems~\cite{DBLP:books/cu/p/Heuvel13,DBLP:journals/algorithms/Nishimura18,DBLP:journals/corr/abs-2204-10526}.

In the \textsc{Token Sliding Reachability (TSR)} problem, introduced by Hearn and Demaine \cite{DBLP:journals/tcs/HearnD05}, two independent sets are adjacent if one can 
be obtained from the other by removing a token from a vertex $u$ and immediately placing it on another unoccupied vertex $v$ with the requirement that $\{u, v\}$ must be an 
edge of the graph. The token is said to \emph{slide} from vertex $u$ to vertex $v$ along the edge $\{u, v\}$. 
Generally speaking, in the \textsc{Token Sliding Reachability} problem, we are given a graph $G$ and two independent sets $S$ and $T$ of size $k$ in $G$. 
The goal is to decide whether there exists a sequence of slides (a \emph{reconfiguration sequence}) that transforms $S$ to $T$.
The \textsc{TSR} problem has been extensively studied~\cite{DBLP:conf/wg/BonamyB17,DBLP:conf/swat/BonsmaKW14,DBLP:conf/isaac/DemaineDFHIOOUY14,
DBLP:conf/isaac/Fox-EpsteinHOU15,DBLP:conf/tamc/ItoKOSUY14,DBLP:journals/tcs/KaminskiMM12,DBLP:journals/jcss/LokshtanovMPRS18}. 
It is known that the problem is PSPACE-complete, even on restricted 
graph classes such as planar graphs of bounded bandwidth (and hence pathwidth)~\cite{DBLP:journals/tcs/HearnD05,DBLP:journals/jcss/Wrochna18, DBLP:conf/iwpec/Zanden15}, split 
graphs~\cite{DBLP:journals/mst/BelmonteKLMOS21}, and bipartite graphs~\cite{DBLP:journals/talg/LokshtanovM19}. However, \textsc{Token Sliding Reachability} can be decided in 
polynomial time on trees~\cite{DBLP:conf/isaac/DemaineDFHIOOUY14}, interval graphs~\cite{DBLP:conf/wg/BonamyB17}, bipartite permutation and bipartite 
distance-hereditary graphs~\cite{DBLP:conf/isaac/Fox-EpsteinHOU15}, line graphs~\cite{DBLP:journals/tcs/ItoDHPSUU11}, and claw-free graphs~\cite{DBLP:conf/swat/BonsmaKW14}.
In the \textsc{Token Sliding Optimization (TSO)} problem, we are additionally given a parameter $\ell$ and the goal is to decide if $S$ can be transformed to $T$ in at most 
$\ell$ token slides. Very little is known about the optimization variant of the problem other than the hardness results that follow immediately from the reachability variant. 
In fact, to the best of our knowledge, the only known polynomial-time solvable instances of \textsc{TSO} are those  
restricted to interval graphs~\cite{DBLP:journals/tcs/YamadaU21,DBLP:journals/tcs/ItoO19}, cographs~\cite{DBLP:journals/tcs/KaminskiMM12}, and spider trees (trees obtained by attaching paths to a central vertex)~\cite{DBLP:conf/ciac/HoangKU19}.  

In the \textsc{Token Jumping Reachability (TJR)} problem, introduced by Kami\'{n}ski et al.~\cite{DBLP:journals/tcs/KaminskiMM12}, we drop the restriction that the token should 
move along an edge of $G$ and instead we allow it to move to any unoccupied vertex of $G$ provided it does not break the independence of the set of tokens. 
That is, a single reconfiguration step consists of first removing a token on some vertex $u$ and then immediately adding it back on any other unoccupied vertex $v$, 
as long as no two tokens become adjacent. The token is said to \emph{jump} from vertex $u$ to vertex $v$.
\textsc{Token Jumping Reachability} is also PSPACE-complete on planar graphs 
of bounded bandwidth~\cite{DBLP:journals/tcs/HearnD05,DBLP:journals/jcss/Wrochna18,DBLP:conf/iwpec/Zanden15}. 
Lokshtanov and Mouawad~\cite{DBLP:journals/talg/LokshtanovM19} showed that, unlike \textsc{Token Sliding Reachability}, which is PSPACE-complete on bipartite graphs, 
the \textsc{Token Jumping Reachability} problem becomes NP-complete on bipartite graphs. 
On the positive side, it is ``easy'' to show that \textsc{Token Jumping Reachability} can be decided in polynomial-time on trees (and even on split/chordal graphs) 
since we can simply jump tokens to leaves (resp. vertices that only appear in the bag of a leaf in the clique tree) to transform one independent set into another. 
In the \textsc{Token Jumping Optimization (TJO)} problem, we are additionally given a parameter $\ell$ and the goal is to decide if $S$ can be transformed to $T$ in at most 
$\ell$ token jumps. To the best of our knowledge, the only known polynomial-time solvable instances of \textsc{TJO} are those restricted to 
even-hole-free graphs~\cite{DBLP:journals/tcs/KaminskiMM12,DBLP:journals/algorithms/MouawadNRS18}. 

In this paper we focus on the parameterized complexity of the aforementioned problems with respect to parameters $k$ and $\ell$ and when restricted to sparse 
classes of graphs. Given an NP-hard or PSPACE-hard problem, \emph{parameterized complexity}~\cite{DBLP:series/mcs/DowneyF99} allows us to refine the notion of hardness; 
does the hardness come from the whole instance 
or from a small parameter? A problem $\Pi$ is \emph{FPT (fixed-parameter tractable) parameterized by $k$} 
if one can solve it in time $f(k) \cdot poly(n)$, for some computable 
function $f$ (sometimes called FPT-time). In other words, the combinatorial explosion can be restricted to the parameter $k$. In the rest of the paper, we mainly consider 
parameters $k$ (the number of tokens) and $\ell$ (the number of reconfiguration steps). 
\textsc{TSO} and \textsc{TJO} are known to be W[1]-hard (and XNL-complete~\cite{bodlaender_et_al:LIPIcs.IPEC.2021.9}) parameterized by $k + \ell$ on general graphs~\cite{DBLP:journals/corr/abs-2204-10526}. 
\textsc{TSR} and \textsc{TJR} are known to be W[1]-hard (and XL-complete~\cite{bodlaender_et_al:LIPIcs.IPEC.2021.9}) parameterized by $k$ on general graphs~\cite{DBLP:journals/jcss/LokshtanovMPRS18}. 
When we restrict our attention to sparse classes of graphs, \textsc{TSO} and \textsc{TJO} are known to be fixed-parameter tractable when parameterized by $\ell$ 
on nowhere dense classes of graphs~\cite{DBLP:journals/algorithms/MouawadNRS18}. 
\textsc{TJR} and \textsc{TJO} are known to be fixed-parameter tractable parameterized by $k$ on graphs of bounded degree~\cite{DBLP:journals/dam/ItoKOSUY20}. 
For \textsc{TJR}, the problem becomes fixed-parameter tractable parameterized by $k$ on biclique-free classes of graphs~\cite{DBLP:conf/fct/BousquetMP17}. 
Finally, for \textsc{TSR}, the problem becomes fixed-parameter tractable parameterized by $k$ on planar graphs, chordal graphs of bounded clique number, and 
graphs of bounded degree~\cite{DBLP:journals/corr/abs-2204-05549}. We refer the reader to the recent survey by Bousquet et al.~\cite{DBLP:journals/corr/abs-2204-10526} 
for more background on the parameterized complexity of these problems. 

Given that \textsc{TSO} and \textsc{TJO} are fixed-parameter tractable when parameterized by $\ell$ 
on nowhere dense classes of graphs, it is natural to ask whether this result can be extended beyond 
nowhere dense graphs to biclique-free graphs. Even simpler, can we show that \textsc{TSO} and \textsc{TJO} 
remain fixed-parameter tractable when parameterized by $\ell$ on graph of bounded degeneracy? Recall that 
any degenerate or nowhere dense class of graphs is a biclique-free class, but not vice versa. Motivated by these questions, 
we show the following:

\begin{itemize} 
\item Both problems are fixed-parameter tractable for parameter $k + \ell + d$ on $d$-degenerate graphs; 
\item Both problems are fixed-parameter tractable for parameter $|N| + k + \ell + d$ on graphs having a modulator $N$ whose deletion leaves a $d$-degenerate graph (assuming $N$ is given as part of the input); and
\item Both problems are fixed-parameter tractable for parameter $|M| + \ell + \Delta$ on graphs having a modulator $M$ whose deletion leaves a graph of maximum degree $\Delta$. 
\item We complement these result by showing that for parameter $\ell$ alone both problems become W[1]-hard already on $2$-degenerate graphs. 
\end{itemize}

In fact, our hardness reductions construct $2$-degenerate graphs which can be partitioned into two sets $V_1$ and $V_2$, where $V_1$ is an independent set 
and every vertex in $V_2$ has constant degree in the graph. Hence, our positive result for parameter $|M| + \ell + \Delta$ shows that when $|M|$ is part of our parameter 
we can drop $k$ and still obtain fixed-parameter tractable algorithms; and when $|M|$ (and $k$) is not part of the parameter the problem is W[1]-hard. 

Most of our positive results make use of the notion of independence covering 
families introduced by Lokshtanov et al.~\cite{DBLP:journals/talg/LokshtanovPSSZ20}, 
which we believe could be of independent interest for the reconfiguration of independent sets. 
Let us start by formally defining such families and the various algorithms for extracting them on different graph classes. 

\begin{definition}[\cite{DBLP:journals/talg/LokshtanovPSSZ20}]\label{def-family}
For a graph $G$ and $k \geq 1$, a family of independent sets of $G$ is called an \emph{independence covering family for $(G, k)$}, denoted by $\mathcal{F}(G, k)$,
if for any independent set $I$ in $G$ of size at most $k$, there exists $J \in \mathcal{F}(G, k)$ such that $I \subseteq J$.
\end{definition}

\begin{theorem}[\cite{DBLP:journals/talg/LokshtanovPSSZ20}]\label{thm-family-degenerate}
There is a deterministic algorithm that given a $d$-degenerate graph $G$ and $k \geq 1$, runs in time
$O((kd)^{O(k)} \cdot (n + m) \log n)$, and outputs an independence covering family for $(G, k)$ of size at most $O((kd)^{O(k)} \cdot \log n)$.
\end{theorem}

\begin{theorem}[\cite{DBLP:journals/talg/LokshtanovPSSZ20}]\label{thm-family-mod}
Let $k,d \in \mathbb{N}$ and $G$ be a graph. Let $S \subseteq V(G)$ such that $G - S$ is $d$-degenerate. 
There is a deterministic algorithm that given a $G$, $S$, and $k,d \in \mathbb{N}$, runs in time
$O(2^{|S|} \cdot (kd)^{O(k)} \cdot 2^{O(kd)} \cdot (n + m) \log n)$, and outputs an independence covering family for $(G, k)$ 
of size at most $O(2^{|S|} \cdot (kd)^{O(k)} \cdot 2^{O(kd)} \cdot \log n)$.
\end{theorem}

\begin{theorem}[\cite{DBLP:journals/talg/LokshtanovPSSZ20}]\label{thm-family-nowhere}
Let $G$ be a graph such that $G \in \mathcal{G}$, where $\mathcal{G}$ is a class of nowhere dense graphs. 
There is a deterministic algorithm that given $k \geq 1$, runs in time
$O(f_{\mathcal{G}}(k) \cdot (n + m) \log n)$, and outputs an independence covering family for $(G, k)$ of size at most $O(g_{\mathcal{G}}(k) \cdot n \log n)$, 
where $f_{\mathcal{G}}(k)$ and $g_{\mathcal{G}}(k)$ depend on $k$ and the class $\mathcal{G}$ but are independent of the size of the graph. 
\end{theorem}

We use Theorems~\ref{thm-family-degenerate} and~\ref{thm-family-mod} to design fixed-parameter tractable algorithms 
for parameters $k + \ell + d$ and $|N| + k + \ell + d$, respectively.
Our algorithm for parameter $|M| + \ell + \Delta$ is based on the random separation technique~\cite{DBLP:conf/iwpec/CaiCC06}. 
Finally, we show that using independence covering families one can obtain a simpler 
and unified algorithm for the standard \textsc{Token Jumping Reachability} problem (a.k.a. \textsc{Token Jumping}) parameterized by $k$ on both 
degenerate and nowhere dense classes of graphs; this is in contrast to the algorithms presented in~\cite{DBLP:journals/jcss/LokshtanovMPRS18}. 
To do so, we make use of Theorems~\ref{thm-family-degenerate} and~\ref{thm-family-nowhere}. 
Note that the major difference between Theorems~\ref{thm-family-degenerate} and~\ref{thm-family-nowhere} is that in the former we are
guaranteed a family of size at most $O((kd)^{O(k)} \cdot \log n)$ while in the latter the family is of size at least $O(g_{\mathcal{G}}(k) \cdot n \log n)$, i.e., we have 
an extra linear dependence on $n$. This difference is the reason why our algorithm for parameter $k + \ell + d$ cannot be adapted to work for nowhere dense graphs. 
The current complexity status of all problems considered in this work is summarized in Table~\ref{tab:complexity-status}. 

\begin{table}[h]
    \centering
    \caption{Parameterized complexity status of the reachability and optimization variants of \textsc{Token Sliding} and \textsc{Token Jumping}. 
	Results proved in this paper are shown in bold.}\label{tab:complexity-status}
    \begin{tabular}{|c|c|c|c|}
    \hline
    \textbf{} & Degenerate & Nowhere dense & Biclique free \\
    \hline
    \textsc{TSR} parameterized by $k$ & Open & Open & Open \\
    \hline
    \textsc{TSO} parameterized by $k$ & Open & Open & Open \\
    \hline
    \textsc{TSO} parameterized by $\ell$ & \textbf{W[1]-hard} & FPT & \textbf{W[1]-hard} \\
    \hline
    \textsc{TSO} parameterized by $k + \ell$ & \textbf{FPT} & FPT & Open \\
    \hline
    \textsc{TJR} parameterized by $k$ & FPT & FPT & FPT \\
    \hline
    \textsc{TJO} parameterized by $k$ & Open & Open & Open \\
    \hline
    \textsc{TJO} parameterized by $\ell$ & \textbf{W[1]-hard} & FPT & \textbf{W[1]-hard} \\
    \hline
    \textsc{TJO} parameterized by $k + \ell$ & \textbf{FPT} & FPT & Open \\
    \hline
    \end{tabular}
\end{table}

The rest of the paper is organized as follows. In Section~\ref{sec:prelims} we introduce required background and terminology. 
In Section~\ref{sec:fpt-ellk-ts} we present our main positive results which are the fixed-parameter tractable algorithms for \textsc{TSO} and \textsc{TJO} 
parameterized by $k + \ell + d$ on $d$-degenerate graphs and parameterized by $|N| + k + \ell + d$ on graphs having a 
modulator $N$ whose deletion leaves a $d$-degenerate graph (assuming $N$ is given as part of the input). 
Section~\ref{sec:fpt-modulator} is devoted to the fixed-parameter tractable algorithm for  
parameter $|M| + \ell + \Delta$ on graphs having a modulator $M$ whose deletion leaves a graph of maximum degree $\Delta$. 
We show hardness on $2$-degenerate graphs in Section~\ref{sec:hard-ts} for \textsc{TSO} and in Section~\ref{sec:hard-tj} for \textsc{TJO}. 
We conlude in Section~\ref{sec:fpt-k-tj} where we present a unified algoritm for \textsc{Token Jumping Reachability} on graphs admitting 
efficiently computable independence covering families of the appropriate size.

\section{Preliminaries}\label{sec:prelims}

\subparagraph{Sets and functions.}
We denote the set of natural numbers (including $0$) by $\mathbb{N}$. For $n \in \mathbb{N}$, we use $[n]$ and $[n]_0$ to denote 
the sets $\{1,2,\cdots,n\}$ and $\{0,1,2,\cdots,n\}$, respectively. For a set $X$, we denote its power set by $2^{X} = \{X' \mid X' \subseteq X\}$. 
For a function $f:X\rightarrow Y$ and an element $y\in Y$, $f^{-1}(y)$ 
denotes the set $\{x\in X\mid f(x)=y\}$. For a non-empty set $X$, a family $\C{F} \subseteq 2^X$ is a {\em partition} of $X$, if i) for 
each $Y \in \C{F}$, $Y \neq \emptyset$, ii) for distinct $Y,Z \in \C{F}$, we have $Y \cap Z = \emptyset$, and iii) $\cup_{Y \in \C{F}} Y = X$.  
An observation that we will make use of is the following:
\begin{proposition}\label{prop-log-fpt}
For $k,n \in \mathbb{N}$, where $n \geq c$ for some constant $c$, we have $(\log n)^k \leq n + k^{2k}$.
\end{proposition}
\begin{proof}
We distinguish between two cases:
\begin{itemize}
\item If $k \leq \log n / \log \log n$ then we have $k \log \log n \leq \log n$. Raising both sides to the power of $2$, we obtain $(\log n)^k \leq n$.
\item If $k > \log n / \log \log n$ then we claim that $\log \log n < k$; suppose otherwise. 
Then we must have, for all $n$, $\log n < k \log \log n \leq (\log \log n)^2$ which is false.
Hence, we have $\log n \leq k^2$ and we get $(\log n)^k \leq k^{2k}$.  
\end{itemize}
Combining the two inequalities we get $(\log n)^k \leq n + k^{2k}$.
\end{proof}

\subparagraph{Graphs and graph classes.}
We assume that each graph $G$ is a
simple, undirected graph with vertex set $V(G)$ and
edge set $E(G)$, where $|V(G)| = n$ and $|E(G)| = m$.
The {\em open neighborhood}, or simply {\em neighborhood}, of a
vertex $v$ is denoted by $N_G(v) = \{u \mid \{u,v\} \in E(G)\}$, the
{\em closed neighborhood} by $N_G[v] = N_G(v) \cup \{v\}$. Similarly, for a set of vertices $S \subseteq V(G)$,
we define $N_G(S) = \{v \mid \{u,v\} \in E(G), u \in S, v \not\in S \}$ and $N_G[S] = N_G(S) \cup S$.
The {\em degree} of a vertex is $|N_G(v)|$.
We drop the subscript $G$ when clear from context.
A {\em subgraph} of $G$ is a graph $G'$
such that $V(G') \subseteq V(G)$ and $E(G') \subseteq E(G)$.
The {\em induced subgraph} of $G$ with respect to $S \subseteq V(G)$ is denoted by $G[S]$;
$G[S]$ has vertex set $S$ and edge set $E(G[S]) = \{\{u,v\} \in E(G) \mid u, v \in S\}$.

{\em Contracting} an edge $\{u,v\}$ of $G$ results in a new graph $H$ in which
the vertices $u$ and $v$ are deleted and replaced by a new vertex $w$ that is adjacent
to $(N_G(u) \cup N_G(v)) \setminus \{u, v\}$.
If a graph $H$ can be obtained from $G$ by repeatedly
contracting edges, $H$ is said to be a {\em contraction} of $G$.
If $H$ is a subgraph of a contraction of $G$, then $H$ is said
to be a {\em minor} of $G$, denoted by $H \preceq_m G$.
The class of nowhere dense graphs~\cite{DBLP:journals/ejc/NesetrilM08,DBLP:journals/ejc/NesetrilM11a} is a common generalization of proper minor
closed classes, classes of graphs with bounded degree, graph classes locally excluding a fixed graph $H$ as a minor and classes of bounded expansion. 
In order to formally define the class of nowhere dense graphs, we need a few additional definitions.

\begin{definition}
A graph $H$ is an \emph{$r$-shallow minor} of $G$, where $r$ is an integer, if there exists a set of disjoint subsets 
$V_1, \ldots, V_{|H|}$ of $V(G)$ such that 
\begin{enumerate}
\item each graph $G[V_i]$ is connected and has radius at most $r$, and
\item there is a bijection $\psi: V(H) \rightarrow \{V_1, \ldots, V_{|H|}\}$ such that for every edge $\{u,v\} \in E(H)$ there
is an edge in $G$ with one endpoint in $\psi(u)$ and the second in $\psi(v)$.
\end{enumerate}
The set of all $r$-shallow minors of a graph $G$ is denoted by $G \triangledown r$. Similarly, the set of all
$r$-shallow minors of all the members of a graph class $\mathscr{C}$ is denoted by $\mathscr{C} \triangledown r = \bigcup_{G \in \mathscr{C}}(G \triangledown r)$.
\end{definition}

\noindent
let $\omega(G)$ denotes the size of the largest clique in $G$ and $\omega(\mathscr{C}) = sup_{G \in \mathscr{C}}(\omega(G))$.

\begin{definition}\label{def:nowhere-dense}
A class of graphs $\mathscr{C}$ is said to be {\em nowhere dense} if there exists a function
$f_\omega : \mathbb{N} \rightarrow \mathbb{N}$ such that for all $r$ we have that $\omega(\mathscr{C} \triangledown r) \leq f_\omega(r)$.
\end{definition}

Nowhere density turns out to be a very robust concept with several
natural characterizations and applications (see, e.g.,~\cite{DBLP:conf/soda/KreutzerRS17}).

\begin{definition}
A class of graphs $\mathscr{C}$ is said to be {\em $d$-degenerate} if
every induced subgraph of any graph $G \in \mathscr{C}$ has a vertex of degree at most $d$. 
$\mathscr{C}$ is said to be {\em degenerate} if it is $d$-degenerate for some $d$.
\end{definition}

Graphs of bounded degeneracy and nowhere dense graphs are incomparable~\cite{DBLP:journals/jacm/GroheKS17}.
In other words, graphs of bounded degeneracy are somewhere dense.
Degeneracy is a hereditary property, hence an induced subgraph of
a $d$-degenerate graph is also $d$-degenerate.
It is well-known that graphs of treewidth at most $d$ are also $d$-degenerate.
Moreover a $d$-degenerate graph cannot
contain $K_{d+1,d+1}$ as a subgraph, which brings us to the
class of biclique-free graphs.
The relationship between bounded degeneracy, nowhere dense,
and $K_{d,d}$-free graphs was shown by Philip et al. and Telle and Villanger~\cite{DBLP:journals/talg/PhilipRS12,DBLP:conf/esa/TelleV12}.

\begin{definition}
A class of graphs $\mathscr{C}$ is said to be {\em $d$-biclique-free}, for some $d > 0$,
if $K_{d,d}$ is not a subgraph of any $G \in \mathscr{C}$. $\mathscr{C}$ is
said to be {\em biclique-free} if it is $d$-biclique-free for some $d$.
\end{definition}

\begin{proposition}[~\cite{DBLP:journals/talg/PhilipRS12,DBLP:conf/esa/TelleV12}]
Any degenerate or nowhere dense class of graphs is biclique-free but not vice-versa.
\end{proposition}

\section{FPT algorithm for parameter $k + \ell + d$}\label{sec:fpt-ellk-ts}
In this section we start by designing a fixed-parameter tractable algorithm for the \textsc{TSO} problem parameterized by $k + \ell + d$ on $d$-degenerate graphs. 
We then show how the algorithm can be adapted for \textsc{TJO} as well as for parameter $|N| + k + \ell + d$ on graphs 
having a modulator $N$ whose deletion leaves a $d$-degenerate graph (assuming $N$ is given as part of the input). 

We let $(G, S, T, k, \ell)$ denote an instance of \textsc{TSO}, where $G$ is $d$-degenerate. 
Moreover, we assume that we have computed in time $O((kd)^{O(k)} \cdot (n + m) \log n)$ an independence covering family 
$\mathcal{F}(G, k)$ for $(G, k)$ of size at most $O((kd)^{O(k)} \cdot \log n)$ (Theorem~\ref{thm-family-degenerate}). 
Without loss of generality, we assume that both $S$ and $T$ belong to $\mathcal{F}(G, k)$; as otherwise we can simply add them. 
Note that if $(G, S, T, k, \ell)$ is a yes-instance then there exists a sequence $\langle I_0, I_1, ..., I_\ell \rangle$ of independent sets such that for 
all $j \in \{0,\ldots,\ell-1\}$ the set $I_j$ is an independent set of size~$k$ in $G$, $I_0 = S$, $I_\ell = T$, and $I_{j+1}$ is obtained from $I_j$ by a token slide. 
This implies that there exists a sequence $\langle J_0, J_1, ..., J_\ell \rangle$ of elements of $\mathcal{F}(G, k)$ such that $J_0 = S$, $J_\ell = T$, and for 
$j \in \{1,\ldots,\ell-1\}$ we have $I_j \subseteq J_j$. In what follows, we assume that we guessed a sequence  $\langle J_0, J_1, ..., J_\ell \rangle$ of elements 
of $\mathcal{F}(G, k)$ such that $J_0 = S$ and $J_\ell = T$. Our goal now is to design an algorithm that either finds a reconfiguration sequence 
$\langle I_0 = S, I_1 \subseteq J_1, ...,I_{\ell - 1} \subseteq J_{\ell-1}, I_\ell = T \rangle$ or determine that no such sequence exists. 

We define a constraint as a pair $(X,b)$ where $X \subseteq V(G)$ and $b$ is a positive integer, called the budget of $X$. 
We denote a set of constraints by $C = \{(X,b), \ldots\}$. 
We say that a constraint $(X,b)$ is satisfied (by $Z$) if $|Z \cap X| = b$, where $Z \subseteq V(G)$. 
We say that a set of constraints $C$ is satisfied (by $Z$) if all pairs $(X,b) \in C$ are satisfied (by $Z$). We denote a set of sets of constraints by $\C{C}$.
We now proceed by building sets of sets of constraints $\C{C}_0$, $\C{C}_1$, $\ldots$, $\C{C}_\ell$ and show that for each $i \in [\ell]_0$, 
the following invariants are satisfied:

\begin{itemize}
\item \textbf{Correctness Invariant I:} If a $k$-sized set $Z \subseteq J_i$ satisfies at least one set of constraints in $\C{C}_i$, then $Z$ is reachable from $S = J_0$. 
\item \textbf{Correctness Invariant II:} For any $k$-sized set $Z \subseteq J_i$, if there is a reconfiguration 
sequence $S = I_0, I_1, I_2, \ldots, I_{i} = Z$, where for each $p \in [i]_0$, $I_p \subseteq J_p$, then $Z$ satisfies at least one set of constraints in $\C{C}_i$. 
\item \textbf{Size Invariant:} The total number of constraints at the $i^{th}$ step is $\sum_{C \in \C{C}_i} |C| \leq (i + 1)!$. 
\end{itemize}

At the base case, we let $\C{C}_0 = \{\{(S,k)\}\}$. The correctness of the base case immediately follows from its construction. 
We now proceed recursively as follows. Consider $i \in [\ell]$. We assume that for each $p \in [i-1]$, we have computed $\C{C}_p$ that 
satisfy the correctness and size invariants. Initialize $\C{C}_i = \emptyset$.

\subparagraph{}
For each $C \in \C{C}_{i-1}$ \par
    For each constraint $(X,b) \in C$
        \begin{enumerate}
            \item Initialize a constraint set $C' = \emptyset$;
            \item If $b = 1$
            \begin{enumerate}
                \item Add $(N(X) \cap J_i,1)$ to $C'$;
                \item Add $(X' \cap J_i, b')$ for all other constraints $(X',b') \in C$ to $C'$;
            \end{enumerate}
            \item Else
            \begin{enumerate}
                \item Add $(X \cap J_i,b-1)$ to $C'$;
                \item Add $(N(X) \cap J_i,1)$ to $C'$;
                \item Add $(X' \cap J_i, b')$ for all other constraints $(X',b') \in C$ to $C'$;
            \end{enumerate}
            \item Add $C'$ to $\C{C}_i$;
        \end{enumerate}

\begin{lemma}\label{lem-subsets}
    For every $C \in \C{C}_i$, $\bigcup_{(X,b) \in C} X \subseteq J_i$.
\end{lemma}

\begin{proof}
We use induction to prove the lemma. For the base case, $i = 0$, we have $\C{C}_0 = \{\{(S,k)\}\}$. 
For $C = \{(S,k)\}$, we can see that the lemma holds. 
For the inductive step, we assume that the lemma holds true for $i - 1$ and prove that it still holds for $i$. 
So, for all $C \in \C{C}_{i-1}$, $\cup_{(X,b) \in C} X \subseteq J_{i-1}$.
In the $i^{th}$ step of the algorithm, we add new sets of constraints $C'$ such that all the constraints $(Y,\beta) \in C'$ have $Y \subseteq J_i$. 
Hence, their union must be a subset of $J_i$. This completes the proof of the lemma.
\end{proof}

\begin{lemma}\label{lem-budgets}
    For every $C \in \C{C}_i$, $\sum_{(X,b) \in C} b = k$.
\end{lemma}

\begin{proof}
We use induction to prove the lemma. For the base case, $i = 0$, we have $\C{C}_0 = \{\{(S,k)\}\}$. For $C = \{(S,k)\}$, we can see that the lemma holds.
For the inductive step, we assume that the lemma holds true for $i-1$, and prove it for $i$. So, for all $C \in \C{C}_{i-1}$, $\sum_{(X,b) \in C} b = k$.
In the $i^{th}$ recursive step of the algorithm, we add a new set of constraints $C'$ corresponding to each constraint $(X,b)$ contained in 
some member of $\C{C}_{i-1}$. If $b$ is $1$, we add another constraint with budget $1$. Otherwise, we split the budget in the previous budget, i.e. $b$, 
into two parts $b-1$ and $1$. The total budget still remains the same as the ${i-1}^{th}$ step, i.e., $k$. This completes the proof of the lemma.
\end{proof}

\begin{lemma}\label{lem-disjoint}
    For every $C \in \C{C}_i$, all the vertex subsets which are part of the constraints in $C$ are pairwise disjoint.
\end{lemma}

\begin{proof}
We use induction to prove the lemma.
For the base case, $i = 0$, we have $\C{C}_0 = \{\{(S,k)\}\}$. For $C = \{(S,k)\}$, we can see that the lemma holds trivially; since $|C| = 1$. 
For $i = 1$, we have $\C{C}_1 = \{\{(S \cap J_1,k-1),(N(S) \cap J_1,1)\}\}$. For $C = \{(S \cap J_1,k-1),(N(S) \cap J_1,1)\}$, we can see that 
$(S \cap J_1) \cap (N(S) \cap J_1) = \emptyset$ and the lemma holds; since $J_1$ is an independent set.
For the inductive step, we assume that the lemma holds true for $i-1$, and prove it for $i$. So, for all $C \in \C{C}_{i-1}$, all the vertex subsets 
which are part of the constraints in $C$ are pairwise disjoint.
In the $i^{th}$ step of the algorithm, we add a new set of constraints $C'$ corresponding to each constraint $(X,b)$ contained 
in some member of $\C{C}_{i-1}$, say $C$. The sets $X' \cap J_i$ added corresponding to all constraints $(X',b') \in C$ such that $(X',b') \neq (X,b)$ are pairwise 
disjoint by the induction hypothesis. If $b > 1$, the set $X \cap J_i$ added is disjoint with all the sets $X' \cap J_i$ such 
that $(X',b') \in C$ and $(X',b') \neq (X,b)$ by the induction hypothesis. 
The set $N(X) \cap J_i$ added is disjoint with $X' \cap J_i$ for all $(X',b') \in C$ because all sets $X'$ are part of an independent 
set $J_i$ by Lemma 3.1, and none of them can have their neighbourhoods intersecting with the other sets. This completes the proof of the lemma.
\end{proof}

\begin{lemma}[Size Invariant]\label{lem-size-inv}
The total number of constraints at the $i^{th}$ step is $c_i = \sum_{C \in \C{C}_i} |C| \leq (i + 1)!$. Therefore, $c_\ell \leq (\ell+1)!$.
\end{lemma}

\begin{proof}
Let $c_i = \sum_{C \in \C{C}_i} |C|$. 
We have $c_0 = 1$ from the base case of the algorithm.
At each step the number of constraints in a set of constraints added increases by at most $1$. 
For $i = 0$, we have only one constraint in $\{(S,k)\} \in \C{C}_0$. Therefore, at the $i^{th}$ step, the maximum number of constraints in 
any set contained in $\C{C}_i$ is at most $i+1$.
In the $i^{th}$ recursive step of the algorithm, we add a new set of constraints $C'$ corresponding to each constraint $(X,b)$ contained in some 
member of $\C{C}_{i-1}$. So, we get the following recursive relation: $|\C{C}_i| = c_{i-1}$. 
Using the fact that all members of $\C{C}_i$ contain at most $i+1$ constraints, we get that $c_i = \sum_{C \in \C{C}_i} |C| \leq (i+1)|\C{C}_i| = (i+1)c_{i-1}$.
Solving the recurrence, we get $c_i \leq (i+1)!$. Therefore, $c_\ell \leq (\ell+1)!$.
\end{proof}

\begin{lemma}[Correctness Invariant I]\label{lem-correct-1}
If a $k$-sized independent set $Z \subseteq J_i$ satisfies at least one set of constraints in $\C{C}_i$ then $Z$ is reachable from $S$.
\end{lemma}

\begin{proof}
We use induction to prove the lemma. 
For $i = 0$, we have $\C{C}_0 = \{\{(S,k)\}\}$. For $C = \{(S,k)\}$, we can see that the lemma holds trivially.
For the inductive step, we assume that the lemma holds true for $i-1$, and prove it for $i$.
Let $Z \subseteq J_i$ and $|Z| = k$ such that it satisfies some set of constraints $C \in \C{C}_i$. Let $(X,b)$ be the constraint in $C' \in \C{C}_{i-1}$ which 
produces this set of constraints $C$ in the $i^{th}$ step of the algorithm.
Since $Z$ satisfies the set of constraints $C$, we have:

\begin{enumerate}
    \item $|Z \cap (N(X) \cap J_i)| = 1$
    \item $|Z \cap (X \cap J_i)| = |Z \cap X| = b - 1$ ($0$ if $b = 1$)
    \item $|Z \cap (X' \cap J_i)| = |Z \cap X'| = b'$ for all other constraints $(X',b') \in C'$
\end{enumerate}

Let $v^*$ be the vertex in $Z \cap (N(X) \cap J_i)$. Let $u^*$ be a vertex in $X$ sharing an edge with $v^*$. 
Take $Z' = (Z \cup \{u^*\}) \setminus \{v^*\}$. It can be seen that $|Z'| = k$ and $Z$ can be obtained from $Z'$ by sliding one token. 
By our construction of $Z'$, it must satisfy the following conditions:

\begin{enumerate}
    \item $|Z' \cap X| = b \geq 1$ (since $u^*$ is included in $Z'$); and
    \item $|Z' \cap X'| = b'$ for all other constraints $(X',b') \in C'$.
\end{enumerate}

Hence, $Z'$ satisfies the set of constraints $C' \in \C{C}_{i-1}$.
Consequently, $|Z' \cap (\cup_{(X',b') \in C'} X')| = \sum_{(X',b') \in C'} |Z' \cap X'| = \sum_{(X',b') \in C'} b' = k$, where the first 
equality follows from the fact that all $X'$ such that $(X',b') \in C'$ are pairwise disjoint by Lemma~\ref{lem-disjoint} and the last equality 
follows from Lemma~\ref{lem-budgets}. Therefore, $Z' \subseteq \cup_{(X',b') \in C'} X' \subseteq J_{i-1}$ by Lemma~\ref{lem-subsets}.

Thus, $Z'$ is a $k$-sized subset of $J_{i-1}$ and satisfies at least one set of constraints in $\C{C}_{i-1}$. 
By the induction hypothesis, $Z'$ is reachable from $S$. Now, since $Z$ is reachable from $Z'$, $Z$ is also reachable from $S$. 
\end{proof}

\begin{lemma}[Correctness Invariant II]\label{lem-correct-2}
For any $k$-sized independent set $Z \subseteq J_i$, if there is a reconfiguration sequence $S = I'_0, I'_1, I'_2, \ldots, I'_{i} = Z$, where 
for each $p \in [i]_0$, $I'_p \subseteq J_p$, then $Z$ satisfies at least one set of constraints in $\C{C}_i$.
\end{lemma}

\begin{proof}
We use induction to prove the lemma.
For $i = 0$, we have $\C{C}_0 = \{\{(S,k)\}\}$. The set $S$ satisfies the set of constraints $C = \{(S,k)\}$ and the lemma holds.

We now assume that the lemma holds true for $i-1$, and prove it for $i$.
Let $C \in \C{C}_{i-1}$ be the set of constraints that $I'_{i-1}$ satisfies. 
So, $|I'_{i-1} \cap (\cup_{(X,b) \in C} X)| = \sum_{(X,b) \in C} |I'_{i-1} \cap X| = \sum_{(X,b) \in C} b = k$, where the first equality follows from the 
fact that all $X$ such that $(X,b) \in C$ are pairwise disjoint by Lemma~\ref{lem-disjoint} and the last equality follows from Lemma~\ref{lem-budgets}. 
Since $|I'_{i-1}| = k$, we have $I'_{i-1} \subseteq \cup_{(X,b) \in C} X$.
In the $i^{th}$ step of the reconfiguration sequence, we slide a token from $I'_{i-1}$ to $I'_{i}$, i.e. from some set $X$ such that $(X,b) \in C$ to its open neighbourhood. 
Consider the set of constraints $C' \in \C{C}_i$ obtained by splitting the constraint $(X,b)$ in the $i^{th}$ recursive step of the algorithm. 
We will show that $I'_{i}$ satisfies $C'$. 
Since $I'_{i-1}$ satisfies the set of constraints $C$, we have $|I'_{i-1} \cap X| = b$ for all other constraints $(X,b) \in C$. 
So, in the $i^{th}$ step we have $|I'_{i} \cap (N(X) \cap J_i)| = |I'_{i} \cap N(X)| = 1$, where the first equality is 
because $I'_{i} \subseteq J_i$ and the second equality is because $I'_{i}$ is obtained from $I'_{i-1}$ by sliding a token from $X$ to its neighbourhood.
We have $|I'_{i} \cap (X \cap J_i)| = |I'_{i} \cap X| = |I'_{i-1} \cap X| - 1 = b - 1$ as one token is moved from $X$. 
When $b = 1$, we get $I'_{i} \cap X = \emptyset$ and this budget constraint is not included in the $i^{th}$ recursive step of the algorithm.
For all other $(X',b') \in C$, we have $|I'_{i} \cap (X' \cap J_i)| = |I'_{i} \cap X'| = |I'_{i-1} \cap X'| = b'$ as none of the tokens 
in any $X'$ are moved in the $i^{th}$ step of the reconfiguration sequence.
Therefore, $I'_{i} = Z$ satisfies all the constraints in $C'$, as needed. 
\end{proof}

We are now ready to prove our first main theorem. 

\begin{theorem}\label{thm-fpt-ellk-ts}
\textsc{Token Sliding Optimization} is fixed-parameter tractable parameterized by $k + \ell + d$ where $d$ denotes the degeneracy of the graph.
\end{theorem}

\begin{proof}
Let $(G, S, T, k, \ell)$ denote an instance of \textsc{TSO}, where $G$ is $d$-degenerate. 
We first compute in time $O((kd)^{O(k)} \cdot (n + m) \log n)$ an independence covering family 
$\mathcal{F}(G, k)$ for $(G, k)$ of size at most $O((kd)^{O(k)} \cdot \log n)$ (by Theorem~\ref{thm-family-degenerate}). 
We then add $S$ and $T$ to $\mathcal{F}(G, k)$ (in case they do not already belong to $\mathcal{F}(G, k)$). 
Next, we ``guess'' a (iterate over every) sequence $\langle J_0, J_1, ..., J_\ell \rangle$ of elements of $\mathcal{F}(G, k)$ such that $J_0 = S$, $J_\ell = T$. 
Note that this guessing can be accomplished in time $O(((kd)^{O(k)} \cdot \log n)^{\ell + 1})$, which by Proposition~\ref{prop-log-fpt} is still FPT-time. 
Moreover, note that, since we are in the token sliding model, no two consecutive sets $J_i$ and $J_{i+1}$ can be identical. 
Finally, we compute $\C{C}_0$, $\C{C}_1$, $\ldots$, $\C{C}_\ell$, which by Lemma~\ref{lem-size-inv} can also be done in FPT-time. 
To conclude, we simply need to check whether $T$ satisfies at least one set of constraints in $\C{C}_{\ell}$. 
The correctness of the algorithm follows from Lemma~\ref{lem-correct-1} and~\ref{lem-correct-2}. 
\end{proof}

\begin{theorem}\label{thm-fpt-ellkmod-ts}
\textsc{Token Sliding Optimization} is fixed-parameter tractable parameterized by $|N| + k + \ell + d$ on graphs having a modulator $N$ whose 
deletion leaves a $d$-degenerate graph (assuming $N$ is given as part of the input).
\end{theorem}

\begin{proof}
We proceed exactly as in the proof of Theorem~\ref{thm-fpt-ellk-ts} but we invoke Theorem~\ref{thm-family-mod} instead of Theorem~\ref{thm-family-degenerate}. 
\end{proof}

We conlude this section by showing how we can adapt the previous two results for the \textsc{Token Jumping Optimization} problem. 

\begin{theorem}
\textsc{Token Jumping Optimization} is fixed-parameter tractable parameterized by $k + \ell + d$ where $d$ denotes the degeneracy of the graph and 
fixed-parameter tractable parameterized by $|N| + k + \ell + d$ on graphs having a modulator $N$ whose 
deletion leaves a $d$-degenerate graph (assuming $N$ is given as part of the input).
\end{theorem}

\begin{proof}
To allow tokens to jump to arbitrary vertices of the graph we only need to slightly modify our construction of the  
sets of sets of constraints $\C{C}_1$, $\ldots$, $\C{C}_\ell$. In particular, we do the following:

\begin{itemize}
	\item For each $C \in \C{C}_{i-1}$
	\begin{itemize}
		\item Initialize a constraint set $C_1$ obtained from $C$ by replacing each $(X, b)$ by $(X \cap J_i, b)$; 
		\item If $X \cap J_i \neq \emptyset$ and $|X| \geq b$ for all $X$ then add $C_1$ to $\C{C}_i$; 
		\item For each constraint $(X,b) \in C$ 
        \begin{enumerate}
            \item Initialize a constraint set $C_2 = \emptyset$;
            \item If $b = 1$
            \begin{enumerate}
                \item Add $((N(X) \cap J_i) \cup (J_i \setminus \cup_{(X',b) \in C} X'), 1)$ to $C_2$;
                \item Add $(X' \cap J_i, b')$ for all other constraints $(X',b') \in C$ to $C_2$;
            \end{enumerate}
            \item Else
            \begin{enumerate}
                \item Add $(X \cap J_i,b-1)$ to $C_2$.
                \item Add $((N(X) \cap J_i) \cup (J_i \setminus \cup_{(X',b) \in C} X'), 1)$ to $C_2$;
                \item Add $(X' \cap J_i, b')$ for all other constraints $(X',b') \in C$ to $C_2$;
            \end{enumerate}
            \item Add $C_2$ to $\C{C}_i$; 
			\item If $|C| > 1$ 
			\begin{enumerate}
				\item Initialize $C' = C \setminus (X,b)$ when $b = 1$ and $C' = (C \setminus (X,b)) \cup (X,b - 1)$ otherwise; 
				\item For each $(X',b') \in C'$ where $X' \neq X$
				\begin{enumerate}
					\item Create a new constraint set $C_3 = (C' \setminus (X',b')) \cup (X',b' + 1)$;
					\item If $X' \cap J_i \neq \emptyset$  and $|X'| \geq b'$ for all $X'$ add $C_3$ to $\C{C}_i$;
				\end{enumerate}
			\end{enumerate}
        \end{enumerate}
	\end{itemize}
\end{itemize}

Note that the constraint set $C_1$ allows us to support token jumps within a same set $X$ while 
the constraint set $C_3$ allows us to support token jumps between two different sets $X$ and $X'$. 
It is not hard to see that one can bound the total number of constraints at the $i^{th}$ step by a function of $k$ and $\ell$ just as in Lemma~\ref{lem-size-inv}. 
Similarly, we can prove that Lemmas~\ref{lem-subsets},~\ref{lem-budgets}, and~\ref{lem-disjoint} still hold, from which we conclude that 
the equivalent of Lemmas~\ref{lem-correct-1} and~\ref{lem-correct-2} also holds, as needed. 
\end{proof}

\section{FPT algorithm for parameter $|M| + \ell + \Delta$}\label{sec:fpt-modulator}
In this section, we prove that \textsc{TSO} and \textsc{TJO} 
are fixed-parameter tractable parameterized by $|M| + \ell + \Delta$. Recall that an instance of either problem 
is denoted by $(G, S, T, k, \ell)$ where $V(G)$ can be partitioned into $H$ and $M$ and every vertex in $H$ has degree at most $\Delta$ in $G$. 
Our algorithm is randomized and based on a variant of the color-coding technique~\cite{DBLP:reference/algo/AlonYZ16} that is particularly
useful in designing parameterized algorithms on graphs of bounded
degree. The technique, known in the literature as random separation~\cite{DBLP:conf/iwpec/CaiCC06}, boils
down to a simple, but fruitful observation that in some cases, if we randomly
color the vertex set of a graph using two colors, the solution or vertices we are
looking for are appropriately colored with high probability. In our case, we want to make sure 
that the set of vertices involved in token slides or jumps gets highlighted. We note that our algorithm is an adaptation of the algorithm 
of Mouawad et al.~\cite{DBLP:journals/algorithms/MouawadNRS18} and it can easily be derandomized using standard techniques~\cite{DBLP:books/sp/CyganFKLMPPS15}. 

We start with an instance $(G = (H,M,E), S, T, k, \ell)$ of \textsc{TSO}; the algorithm is identical for \textsc{TJO}. 
We color independently every vertex of $H$ using one of two colors, say red and blue (denoted by $\mc{R}$ and $\mc{B}$), with probability ${1 \over 2}$. 
We let $\chi: H \rightarrow \{\mc{R}, \mc{B}\}$ denote the resulting random coloring. 
Suppose that $(G,S,T,k,\ell)$ is a yes-instance, and let $\sigma$ denote a reconfiguration sequence 
from $S$ to $T$ of length at most $\ell$. We say a vertex $v \in H$ is \emph{touched} in $\sigma$ whenever a token slides from a neighbor of $v$ to $v$ or 
from $v$ to some neighbor of $v$.  
We let $V(\sigma)$ denote the set of vertices touched by $\sigma$.  
We say that the coloring $\chi$ is {successful} if both of the following conditions hold: 

\begin{itemize}
\item Every vertex in $V(\sigma) \cap H$ is colored red; and 
\item Every vertex in $N_H(V(\sigma) \cap H)$ is colored blue. 
\end{itemize}

Observe that $V(\sigma) \cap H$ and $N_H(V(\sigma) \cap H)$ are disjoint. Therefore, the two aforementioned 
conditions are independent. Moreover, since the maximum degree of $G[H]$ is $\Delta$, we have $|V(\sigma) \cap H| + |N_H(V(\sigma) \cap H)| \leq 2\ell\Delta$. 
Consequently, the probability that $\chi$ is successful is at least: 
\begin{align}
\frac{1}{2^{|V(\sigma) \cap H| + |N_H(V(\sigma) \cap H)|}} \geq \frac{1}{2^{2\ell\Delta}} = \frac{1}{4^{\ell\Delta}}.\notag
\end{align}

Let $H_{\mc{R}}$ denote the set of vertices of $H$ colored red and $H_{\mc{B}}$ denote the set 
of vertices of $H$ colored blue. Moreover, we let $C_1$, $\dots$, $C_q$ denote the set of connected components of $G[H_{\mc{R}}]$. 
The main observation now is the following: 

\begin{lemma}\label{prop-colouring}
If $\chi$ is successful then $V(\sigma)$ has a non-empty intersection with at most $2\ell$ connected components of $G[H_{\mc{R}}]$, 
and each one of those components consists of at most $2\ell$ vertices. 
\end{lemma}

\begin{proof} 
Since $|V(\sigma)| \leq 2\ell$, we know that $G[(V(\sigma) \cup N_G(V(\sigma))) \cap H]$ consists of at most $2\ell$ connected components 
(each of size at most $2\ell + 2\ell\Delta$) and $G[V(\sigma) \cap H]$ consists of at most $2\ell$ components (each of size at most $2\ell$). 

Let $C$ be a connected component of $G[H_{\mc{R}}]$ such that $|V(C)| > 2\ell$. Suppose to the contrary that $V(\sigma) \cap V(C) = Q \neq \emptyset$. 
Since $\chi$ is successful, it must be the case that every vertex in $N_H(Q)$ is colored blue. However, we know that there exists 
at least one vertex in $N_H(Q)$ that is colored red (since $C$ is a connected component of $G[H_{\mc{R}}]$ and all at least $2\ell + 1$ vertices in $C$ are colored red). 
As we have obtained a contradiction, we can conclude that when $\chi$ is successful, $V(\sigma)$ can intersect at most $2\ell$ connected components of 
$G[H_{\mc{R}}]$, and none of those components can be of a size greater than $2\ell$, as needed.
\end{proof}

Given an instance $(G = (H,M,E), S, T, k, \ell)$ of \textsc{TSO} and a coloring $\chi$ of $H$, we know from 
Lemma~\ref{prop-colouring} that when $\chi$ is successful every connected 
component of $G[H_{\mc{R}}]$ consists of at most $2\ell$ vertices. We now construct a new (reduced) instance 
$(G', S', T', k', \ell)$ of \textsc{TSO}. We first guess the vertices of $M$ that will be touched in a solution and we let $M'$ denote this set. 
Note that this guessing can be accomplished in time $2^{M}$-time. 
Starting from a copy of $G$ we proceed as follows: 
\begin{itemize}
\item If there exists $v \in (S \cap T) \cap H$ and $v$ is colored blue then we delete $v$ and its neighbors from the graph (since $v$ will not be touched and its neighbors cannot contain tokens);
\item If there exists $v \in (S \cap T) \cap (M \setminus M')$ then we delete $v$ and its neighbors from the graph (since $v$ will not be touched and its neighbors cannot contain tokens);
\item If there exists $v \in (S \cap T) \cap H$, $v$ is colored red, and $v$ belongs to a red component $C$ of $G[H_{\mc{R}}]$ such that $|V(C)| > 2\ell$ 
then we delete $v$ and its neighbors from the graph (since $v$ will not be touched and its neighbors cannot contain tokens);
\item If there exists a blue vertex $v$ which is not in $S \cap T$ then we delete $v$ from the graph (note that $v$ cannot be in $(S \setminus T) \cup (T \setminus S)$ as otherwise the coloring cannot be successful);
\item If there exists a red vertex  $v$ which is not in $S \cap T$ and $v$ belongs to a red component $C$ of $G[H_{\mc{R}}]$ such that $|V(C)| > 2\ell$ 
then we delete $v$ from the graph ($v$ cannot be in $(S \setminus T) \cup (T \setminus S)$).
\end{itemize}
We adjust $S$, $T$, and $k$ appropriately to obtain the new equivalent instance $(G', S', T', k', \ell)$. 
Note that in this new instance (assuming a successful coloring) no vertices are colored 
blue and (assuming a correct guess) all vertices of $M'$ will be touched in a solution. 
In other words, $G'$ can be partitioned into $M'$ and $H'$ where $H'$ consists of (an unbounded number of) connected components each 
consisting of at most $2\ell$ vertices. Note that when the number of connected components is constant then we are done since 
we can solve the problem via brute-force. In other words, we can simply enumerate all possible sequences of length at most $\ell$ and make sure 
that at least one of them is the required reconfiguration sequence from $S'$ to $T'$. 
This brute-force testing can be accomplished in time $2^{\Oh(\ell \log \ell)} \cdot n^{\Oh(1)}$. 

Let us now consider the general case when the number of components is not necessarily bounded. 
We say a component $C$ of $G' - M'$ is \emph{important} if $V(C) \cap ((S' \setminus T') \cup (T' \setminus S')) \neq \emptyset$. 
There are at most $2\ell$ important components. Hence, we only need to bound the number of \emph{unimportant} components. 
To that end, we partition the unimportant components of $G' - M'$ into equivalence classes with 
respect to the relation $\simeq$. For two graphs $G_1$, $G_2$ and two sets $X_1 \subseteq V(G)$, $X_2 \subseteq V(G_2)$, we 
say that $(G_1,X_1)$ and $(G_2,X_2)$ are \emph{isomorphic} if the graphs $G_1$ and $G_2$ are isomorphic where vertices of $X_1$ and $X_2$ are now assigned the same color. 
Formally, a \emph{$c$-colored graph} $G$ is a tuple $(V,E,\mathcal{K})$ such that $\mathcal{K} = \{K_1, \ldots, K_c\}$ is a collection of subsets of $V(G)$ 
where each $K_i$ is called a \emph{color set}. Two colored graphs $G_1 = (V_1,E_1,\mathcal{K}_1)$ and $G_2 = (V_2,E_2,\mathcal{K}_2)$ 
are isomorphic if there is a \emph{color-preserving isomorphism} $f : V_1(G_1) \rightarrow V_2(G_2)$ such that:
\begin{itemize}
\item for all $u, v \in V_1(G_1)$, $\{u,v\} \in E_1(G_1)$ if and only if $\{f(u), f(v)\} \in E_2(G_2)$; and
\item for all $v \in V_1(G_1)$ and $K^1_i \in \mathcal{K}_1$, $v \in K^1_i$ if and only if $f(v) \in K^2_i$.
\end{itemize}
Hence, $(G_1,X_1)$ and $(G_2,X_2)$ are isomorphic if the colored graphs $G_1 = (V_1, E_1, \{X_1\})$ and $G_2 = (V_2, E_2, \{X_2\})$ are isomorphic. 
Let $C_1$ and $C_2$ be two components in $G' - M'$ and let $N_1$ and $N_2$ be their respective neighborhoods in $M'$. 
We say $C_1$ and $C_2$ are \emph{equivalent}, i.e., $C_1 \simeq C_2$, whenever $N_1 = N_2 = N$ and 
$(G[V(C_1) \cup N], V(C_1) \cap S' \cap T')$ is isomorphic to $(G[V(C_2) \cup N], V(C_2) \cap S' \cap T')$ by an isomorphism that fixes $N$ point-wise. 

\begin{proposition}\label{cl:num-classes2}
The total number of $2$-colored graphs with at most $2\ell$ vertices is at most $2^{\Oh(\ell^2)}$, and therefore, 
the equivalence relation $\simeq$ has at most $2^{\Oh(\ell^2)}$ equivalence classes. 
\end{proposition}

Assume that some equivalence class contains more than $2\ell$ unimportant components. 
We claim that retaining only $2\ell$ of them is enough. 
To see why, it is enough to note that $V(\sigma)$ intersects with at most $2\ell$ of those components; they are all equivalent. 
Putting it all together, we know that we have at most 
$2^{\Oh(\ell^2)}$ equivalence classes, each with at most $2\ell$ components, and each component is of size at most $2\ell$. 
Hence, we can guess the sequence from $S'$ to $T'$ in time $2^{\Oh(\ell^3 \log \ell)} \cdot n^{\Oh(1)}$ 
(testing whether two graphs with $2\ell$ vertices are isomorphic can be accomplished naively in time $2^{\Oh(\ell \log \ell)}$). 

We have proven that the probability that $\chi$ is successful is at least $4^{-\ell\Delta}$. 
Hence, to obtain a Monte Carlo algorithm with false negatives, we repeat the
above procedure $4^{\ell\Delta}$ times and obtain the following result:

\begin{theorem}
There exists a one-sided error Monte Carlo algorithm with false negatives that 
solves \textsc{TSO} and \textsc{TJO} parameterized by $|M| + \ell + \Delta$ in time $\Oh(2^{M} \cdot 4^{\ell\Delta} \cdot 2^{\Oh(\ell^3 \log \ell)} \cdot n^{\Oh(1)})$. 
\end{theorem}

\section{Hardness of \textsc{TSO} parameterized by $\ell$ on $2$-degenerate graphs}\label{sec:hard-ts}
In the \textsc{Multicolored Clique} problem, we are given an input graph $G$ whose vertices are colored with $k$ 
colors and the goal is to find a clique containing one vertex from each color. 
We show that \textsc{TSO} parameterized by $\ell$ is \WOH~on $2$-degenerate graphs via a reduction from \textsc{Multicolored Clique}, known to be \WOH.

We construct an instance $(G', S, T, \kappa, \ell = 8{k\choose 2}+2k)$ of \textsc{TSO} from an instance of \textsc{Multicolored Clique} denoted by 
$(G,k,(V_1,V_2,\ldots,V_k))$, where, w.l.o.g., we assume that there are no edges between two vertices of $G$ of the same color.

\subparagraph*{Construction of $G'$.} 
We subdivide all the edges in $G$. Let the vertex set of $G$ be $V$. All the vertices corresponding to the edges in $G$ are partitioned 
into ${k\choose 2}$ sets of the form $\C{E}_{ij}$, where $i=\{1,2,\ldots,k\}$ and $j=\{1,2,\ldots,k\}$ and $i\ne j$, such that $\C{E}_{ij}$ contains 
all the vertices corresponding to the edges in $G$ having one incident vertex of color $i$ and the other incident vertex of color $j$. 
Let the union of all the sets $\C{E}_{ij}$ be denoted by $E$.

We introduce two independent sets $X$ and $Y$, each of size ${k\choose 2}$. Let us label the vertices in $X$ from $1$ to ${k\choose 2}$ and 
the sets $\C{E}_{ij}$ from $1$ to ${k\choose 2}$. We add edges between vertex with label $b$ in $X$ and all the vertices in the $\C{E}_{ij}$ with label $b$. 
Similarly, we label the vertices of $Y$ and add edges from each vertex in $Y$ to all the vertices in the $\C{E}_{ij}$ having the same label. 
Each of these edges is further subdivided three times. Let the vertices on the subdivided edges from $X$ to $E$, which are 
neither adjacent to some vertex in $X$ nor $E$ be denoted by $U_1$, and the vertices on the subdivided edges from $Y$ to $E$, which are neither 
adjacent to some vertex in $Y$ nor $E$ be denoted by $U_2$. We take $U = U_1 \cup U_2$.

\begin{figure}
\centering
\includegraphics[scale=0.5]{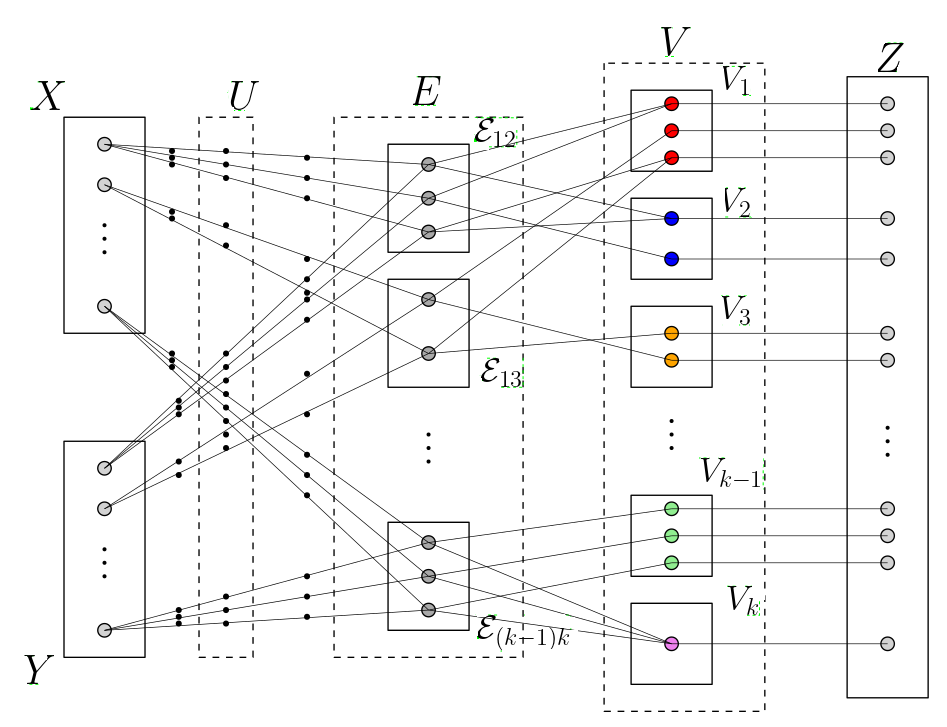}
\caption{An illustration of the reduction from $(G,k,(V_1,V_2,\ldots,V_k))$ to $(G', S, T, \kappa, \ell = 8{k\choose 2}+2k)$.}
\label{fig:replacets}
\end{figure}

We also add a vertex corresponding to each vertex in $V$ and add an edge between the two. 
Let this set of vertices be $Z$. The induced subgraph of $G'$ having $V \cup Z$ as its vertex set forms a perfect matching.
Our initial independent set $S = V \cup X \cup U$ and our target independent set $T = V \cup Y \cup U$. 
Note that $|S| = |T| = n + {k\choose 2} + |U| = \kappa$. We set $\ell = 8{k\choose 2}+2k$. 

\begin{lemma}\label{lem:ts-degenerate}
The graph $G'$ is $2$-degenerate.
\end{lemma}

\begin{proof}
Recall that a graph $G'$ is $2$-degenerate if every induced subgraph $H$ of $G'$ has a vertex of degree at most $2$. 
Consider any induced subgraph $H$ of $G'$. If $H$ contains a vertex of $Z$ or a vertex from the subdivided edges from $X \cup Y$ to $E$ then we are done; 
as those vertices have degree at most two in $G'$. Otherwise, we know that $H$ either contains an isolated vertex from $X \cup Y$ or a degree-two vertex from $E$, as needed. 
\end{proof}

\begin{lemma}\label{lem:l1}
If $(G,k,(V_1,V_2,\ldots,V_k))$ is a yes-instance of \textsc{Multicolored Clique} then there is a reconfiguration sequence of length at most $\ell$ from $S$ to $T$ in $G'$.
\end{lemma}

\begin{proof}
Let the solution to the \textsc{Multicolored Clique} instance be $\{v_1,v_2,\ldots,v_k\} \subseteq V$. 
Consider the following reconfiguration sequence from $S$ to $T$:
\begin{enumerate}
	\item Slide each token on $v_i$ to its matched neighbour in $Z$; for a total of $k$ slides.
	\item Since the vertices $\{v_1,v_2,\ldots,v_k\}$ form a clique in $G$, there are ${k\choose 2}$ edges, each having distinct pair of colors on their incident vertices. 
	So in $G'$, all the vertices corresponding to the edges of the clique lie in distinct partitions $\C{E}_{ij}$. 
	We slide all the tokens from $X$ to $Y$ using these ${k\choose 2}$ vertices. 
	Consider the path from a vertex $v_x \in X$ to a vertex $v_y \in Y$, passing through one of these ${k\choose 2}$ vertices, say $v_i$ where $i \in [k]$. 
	This path contains a vertex $u_1 \in U_1$ and a vertex $u_2 \in U_2$. Slide the token on $u_2$ to $v_y$ ($2$ slides), the 
	token on $u_1$ to $u_2$ through $v_i$ ($4$ slides), and the token on $v_x$ to $u_1$ along this path ($2$ slides); for a total of $8{k\choose 2}$ slides.
    \item Finally we slide the tokens in $Z$ back to $V$; for a total of $k$ slides.
\end{enumerate}
The length of the reconfiguration sequence is $8{k\choose 2}+2k$. This completes the proof.
\end{proof}

\begin{lemma}\label{lem:l2}
If there is a reconfiguration sequence of length at most $\ell$ from $S$ to $T$ in $G'$ then $(G,k,(V_1,V_2,\ldots,V_k))$ is a yes-instance of \textsc{Multicolored Clique}.
\end{lemma}

\begin{proof}
Let the reconfiguration sequence be $I_0,I_1,I_2,\ldots,I_\ell$ where $I_0 = S$, $I_\ell = T$, and $\ell \leq 8{k\choose 2}+2k$.
    
We need at least one step for moving out each token in $X$. 
This requires a total of ${k\choose 2}$ slides. In order to move out a token from $X$, we need to move out a token on $U_1$ on at least one of the 
paths connecting that vertex in $X$ to $\C{E}_{ij}$. This again requires at least ${k\choose 2}$ slides. The tokens moved out from $U_1$ need to be replaced, which 
requires at least ${k\choose 2}$ slides. 
    
Since in the initial configuration $S$ all tokens are at a distance of at least $2$ from the vertices in $Y$, we need at least $2$ slides to bring a token into a vertex in $Y$. 
This amounts to a total of $2{k\choose 2}$ slides. Now, consider the vertex $v$ from which a token is moved to a vertex in $Y$. 
The neighbour of $v$, say $u$, originally had a token, which must have been moved out for placing a token on $v$. This token needs to be replaced by moving in a token 
from an adjacent vertex in $N(E)$, say $v'$. Bringing a token to $v'$ and then moving it to $u$ requires at least $2$ slides. 
So, we get a total of at least $2{k\choose 2}$ slides.
    
The only way to move tokens out of $X \cup N[U_1]$ is through $E$. 
Thus, ${k\choose 2}$ tokens have to be moved out of $X \cup N[U_1]$ to $E$. Every $\C{E}_{ij}$ has at least one token at some point of time. 
This requires another ${k\choose 2}$ slides. In total we have taken up at least $8{k\choose 2}$ slides. So, we are left with a budget of at most $2k$.
    
Any token moved out of $V$ either needs to be brought back or replaced. Both of these require $2$ slides at least. So, we can move out at most $k$ tokens from $V$.
    
Whenever a token is to be moved out of $X \cup N[U_1]$ to a vertex $v_e$ (corresponding to an edge $e$ in $G$) in $\C{E}_{ij}$ for some $i$ and $j$, the tokens 
on the $2$ vertices incident on $e$ must be moved out of $V$. We need to move out all the ${k\choose 2}$ tokens from $X \cup N[U_1]$, one token through each 
of the ${k\choose 2}$ sets $\C{E}_{ij}$. Therefore, we should move out the tokens from $V$ which are adjacent to the vertices in $\C{E}_{ij}$ sets to which 
the tokens from $X \cup N[U_1]$ are moved. Thus for the ${k\choose 2}$ vertices in $E$ (one in each $\C{E}_{ij}$), we can have at most $k$ neighbours in $V$.
    
Let us consider the subgraph induced in $G$ by these set of vertices in $V$. It has at most $k$ vertices and exactly ${k\choose 2}$ edges. 
Now a graph having ${k\choose 2}$ edges must have at least $k$ vertices. So, the induced subgraph has exactly $k$ vertices and forms a clique such that every 
edge has a distinct pair of colors on their incident vertices. This set of vertices in $V$ give us a solution to the \textsc{Multicolored Clique} instance.
\end{proof}

The combination of~\Cref{lem:l1,lem:l2} give us the following: 

\begin{theorem}
\textsc{Token Sliding Optimization} parameterized by $\ell$ is \WOH~on $2$-degenerate graphs.
\end{theorem}

\section{Hardness of \textsc{TJO} parameterized by $\ell$ on $2$-degenerate graphs}\label{sec:hard-tj}

We now show that \textsc{TJO} parameterized by $\ell$ is \WOH~on $2$-degenerate graphs via a reduction from the \textsc{Clique} problem, known to be \WOH.
We construct an instance $(G', S, T, \kappa, 2{k\choose 2}+{k\choose 2}^2+2k)$ of \textsc{TJO} starting from a \textsc{Clique} instance $(G,k)$. 
The construction is quite similar to that of the sliding variant but with some adaptation to account for the possibility of tokens jumping anywhere in the graph. 

\subparagraph*{Construction of $G'$.} 
We subdivide all the edges in $G$. Let the vertex set of $G$ be $V$. We let the set of vertices in $G'$ corresponding to the edges be denoted by $E$.
We introduce a biclique with parts $L$ and $R$, each of size ${k\choose 2}$. 
Next, we subdivide all the edges of the biclique twice. Let this entire set of vertices, i.e. $L \cup N(L) \cup N(R) \cup R$ be denoted by $X$. 
The vertices in $X$ do not have edges with those in $E$ or $V$. 
We also add a vertex corresponding to each vertex in $V$ and add an edge between the two.
Let this set of vertices be $Z$. The induced subgraph of $G'$ having $V\cup Z$ as its vertex set forms a perfect matching.
Our initial independent set $S = V \cup L \cup N(R)$ and our target independent set $T = V \cup R \cup N(L)$. 
Note that $|S| = |T| = \kappa$. 
We let $\ell = 2{k\choose 2} + {k\choose 2}^2 + 2k$.

\begin{figure}
\centering
\includegraphics[scale=0.5]{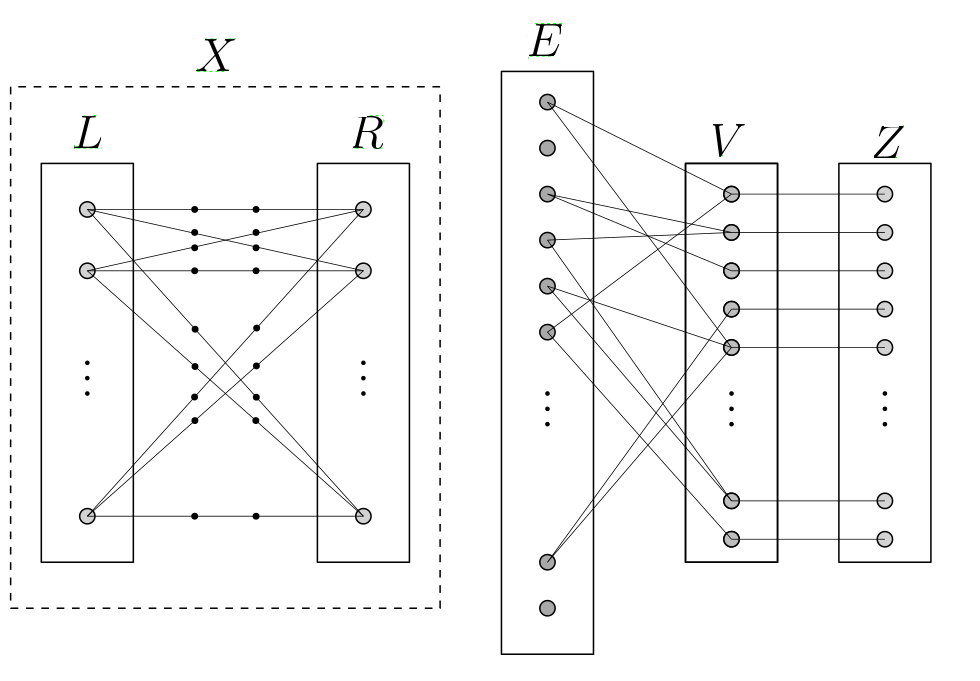}
\caption{An illustration of the reduction from $(G,k)$ to $(G', S, T, \kappa, 2{k\choose 2}+{k\choose 2}^2+2k)$.}
\label{fig:replacetj}
\end{figure}

\begin{lemma}\label{lem:tj-degenerate}
The graph $G'$ is $2$-degenerate.
\end{lemma}

\begin{proof}
Consider any induced subgraph $H$ of $G'$. If $H$ contains a vertex of $Z$ or a vertex from $N(L) \cup N(R)$ then we are done; 
as those vertices have degree at most two in $G'$. Otherwise, we know that $H$ either contains an isolated vertex from $L \cup R$ or a degree-two vertex from $E$, as needed. 
\end{proof}

\begin{lemma}\label{lem:l3}
If $(G,k)$ is a yes-instance of \textsc{Clique} then there is a reconfiguration sequence of length at most $\ell$ from $S$ to $T$ in $G'$.
\end{lemma}

\begin{proof}
Let the solution to the \textsc{Clique} instance be $\{v_1,v_2,\ldots,v_k\} \subseteq V$. 
Consider the following reconfiguration sequence from $S$ to $T$:
\begin{enumerate}
\item Jump each token on $v_i$ to its matched neighbour in $Z$; for a total of $k$ jumps.
\item Since the vertices $\{v_1,v_2,\ldots,v_k\}$ form a clique in $G$, there are ${k\choose 2}$ edges in the subgraph induced on those vertices. 
Let the set of corresponding vertices in $E$ be $E_\C{C}$. We jump all the ${k\choose 2}$ tokens from $L$ to the vertices in $E_\C{C}$; for a total of ${k\choose 2}$ jumps.
\item Jump all the tokens in $N(R)$ to their adjacent vertex in $N(L)$; for a total of ${k\choose 2}^2$ jumps.
\item Now jump all the tokens in $E_\C{C}$ to $R$; for a total of ${k\choose 2}$ jumps.
\item Finally we jump the tokens in $Z$ back to $V$; for a total of $k$ jumps.
\end{enumerate}
The length of the reconfiguration sequence is $2{k\choose 2}+{k\choose 2}^2+2k$. This completes the proof.
\end{proof}

\begin{lemma}\label{lem:l5}
The first time a token jumps to $R$, there can be at most ${k\choose 2}^2$ tokens in the structure $X$.
\end{lemma}

\begin{proof}
Let us assume that the first time a token jumps to $R$ there are ${k\choose 2}^2+1$ tokens in $X$. 
Also, let the vertex in $R$ where a token is to be moved be $v$. If $y > 0$ of these tokens are in $L$ and none of them are in $R$ then 
there cannot be any tokens on the vertices of the subdivided edges from those $y$ vertices in $L$ to $v$. So, we have at most $y$ tokens on the vertices in 
$L$ and those on the paths from them to $v$. We have ${k\choose 2}^2-y$ paths from the remaining vertices in $L$ to the vertices in $R\setminus\{v\}$. 
These paths can have at most ${k\choose 2}^2-y$ tokens. Thus, in total we can have at most $y+{k\choose 2}^2-y={k\choose 2}^2$ tokens in $X$. 
This leads us to a contradiction, which completes the proof. 
\end{proof}

\begin{lemma}\label{lem:l4}
If there is a reconfiguration sequence of length at most $\ell$ from $S$ to $T$ in $G'$ then $(G,k)$ is a yes-instance of \textsc{Clique}.
\end{lemma}

\begin{proof}
Let the reconfiguration sequence be $I_0,I_1,I_2,\ldots,I_\ell$, where $I_0 = S$,  $I_\ell = T$, and $\ell \leq 2{k\choose 2}+{k\choose 2}^2+2k$. 
None of the tokens in $X$ have the same position in both $S$ and $T$. Hence, all of them have to jump at least once. 
This accounts for ${k\choose 2}+{k\choose 2}^2$ steps. From~\Cref{lem:l5}, we know that no tokens can be moved into $R$ until we have no more 
than ${k\choose 2}^2$ tokens left in $X$.
This implies that at least ${k\choose 2}$ of the tokens in $X$ have to be moved out to either $E$, $V$, or $Z$. 
When we are about to move a token into $R$ for the first time, we can have at most ${k\choose 2}^2$ tokens in $X$. 
So, an extra ${k\choose 2}$ tokens have to be moved into $X$, which takes at least ${k\choose 2}$ steps. 
Therefore, we are left with a budget of at most $2k$. We consider the following three cases while jumping tokens out of $X$:

\begin{itemize}
\item {\bf Case I: } If a token from $X$ jumps to some vertex $v_e$ in $E$, the tokens on the two neighbouring vertices of $v_e$ in $V$ should 
have been moved out to $E$ or $Z$ (we are not considering $X$, as moving a token from $V$ to $X$ in order to shift a token out of $X$ does not help).
    
\item {\bf Case II: } If a token from $X$ is to be moved to a vertex $v_z$ in $Z$, the token on the neighbouring vertex of $v_z$ in $V$ needs to be moved out to $E$ or $Z$. 
If that is to be moved to some vertex in $Z$, then the token on the neighbour of the matched vertex in $V$ needs to be jumped out of $V$. 
Again if that token is to be moved out to $Z$, the sequence continues. At most $n-1$ tokens can be jumped out of $V$ to $Z$, because the initial token 
from $X$ would occupy one vertex in $Z$. So, after at most $n-1$ steps in the sequence, the token must be moved out of $V$ to some vertex $v_e$ in $E$. 
Now, the tokens on the two neighbours of $v_e$ in $V$ must have been moved out of $V$ prior to the above sequence of jumps. 
Thus, we can consider a shorter reconfiguration sequence where the token from $X$ is directly jumped to $v_e$ after shifting the tokens on its neighbours in $V$. 
Then it becomes similar to Case I.
    
\item {\bf Case III: } If a token from $X$ is to be moved to a vertex $v_v$ in $V$, the token on $v_v$ needs to be moved out to $E\setminus N(v_v)$ or $Z\setminus N(v_v)$. 
If that is to be moved to some vertex in $Z\setminus N(v_v)$, then the token on the neighbour of the matched vertex in $V$ needs to be jumped out of $V$. 
Again if that token is to be moved out to $Z$, the sequence continues. At most $n-1$ tokens can be jumped out of $V$ to $Z$, because we cannot place a token 
on the neighbour of $v_v$ in $Z$. So, after at most $n-1$ steps in the sequence, the token must be moved out of $V$ to some vertex $v_e$ in $E$. 
Now, the tokens on the two neighbours of $v_e$ in $V$ must have been moved out of $V$ prior to the above sequence of jumps. Thus, we can consider a 
shorter reconfiguration sequence where the token from $X$ is directly jumped to $v_e$ after shifting the tokens on its neighbours in $V$. 
Thus, it becomes similar to Case I.
\end{itemize}
    
By the above case analysis, it is sufficient to consider that Case I holds for our reconfiguration sequence. 
We need to move out at least ${k\choose 2}$ tokens from $X$ through the vertices in $E$. 
Therefore, we should move out the tokens from $V$ which are adjacent to the vertices in $E$ to which the tokens from $X$ are moved. 
While we are moving tokens out of $X$, any token moved out of $V$ eventually needs to be moved to $X$ or brought back to $V$. 
Both of these take at least $2$ steps because we directly cannot jump a token from $V$ to $X$ until $X$ contains at most ${k\choose 2}^2$ tokens. 
So, we can move out at most $k$ tokens from $V$. Thus for the ${k\choose 2}$ vertices in $E$, we can have at most $k$ neighbours in $V$.
    
Let us consider the subgraph in $G$ induced by this set of vertices in $V$. It has at most $k$ vertices and exactly ${k\choose 2}$ edges. 
Now a graph having ${k\choose 2}$ edges must have at least $k$ vertices. So, the induced subgraph has exactly $k$ vertices and forms a clique. 
This set of vertices in $V$ give us a solution to the \textsc{Clique} instance, as needed.
\end{proof}

The combination of~\Cref{lem:l3,lem:l4,lem:l5} give us the follwing:

\begin{theorem}
\textsc{Token Jumping Optimization} parameterized by $\ell$ is \WOH~on $2$-degenerate graphs.
\end{theorem}

\section{FPT algorithm for \textsc{TJR} parameterized by $k$}\label{sec:fpt-k-tj}
	
We propose a generalized scheme for solving \textsc{Token Jumping Reachability} parameterized by $k$ on graphs having a small $k$-independence covering family, i.e., a family 
of size $\C{O}(f(k)\cdot \text{poly}(n))$. Degenerate and nowhere dense graphs admit such independence covering families as shown in~\cite{DBLP:journals/talg/LokshtanovPSSZ20}. 

We remove all the sets in the covering family of size less than $k$. We find out if the independent sets $S$ and $T$ are a part of the 
independence covering family. If not, we add them to the family. Let the size of the resulting $k$-independence covering family $\C{F}(G,k)$ be $q$. 
For denoting an independent set in the family, we will use capital letters like, $X, Y (\subseteq V (G))$. 
We construct a graph $\C{G}$ with $q$ vertices corresponding to the $q$ sets in the family. 
Consider two independent sets $I$ and $I'$ in $\C{F}(G,k)$. We add an edge between the vertices $i$ and $i'$ in $\C{G}$ if and only if $|I \cap I'|\geq k-1$. 
Note that for any two $k$-sized independent sets $J$ and $J'$ we can find a trivial reconfiguration sequence from $J$ to $J'$ if both of them are contained in some 
$I$ in $\C{F}(G,k)$. 

In the algorithm, we find out if the vertices $i_s$ and $i_t$ in $\C{G}$ are in the same connected component. 
If yes, then we know that $S$ is reachable from $T$ from the construction of $\C{G}$. Otherwise no reconfiguration sequence from $S$ to $T$ exists.

%
%
%
%

\begin{lemma}\label{lem:reach1}
If there exists a path from $i_s$ to $i_t$ in $\C{G}$ then there is a reconfiguration sequence from $S$ to $T$ in $G$.
\end{lemma}

\begin{proof}
Let $i_s=i_0, i_1, i_2, \ldots, i_\ell=i_t$ be the path from $i_s$ to $i_t$. 
We start the reconfiguration sequence with $S$. For each pair of vertices $i_j$ and $i_{j+1}$ in the path, we have $|I_j\cap I_{j+1}|\geq k-1$ according to the construction. 
Now, let $X_j \subseteq I_j$ be a $k$-sized independent set in the reconfiguration sequence and $Y_j \subseteq I_j\cap I_{j+1}$ be a $(k-1)$-sized set. 
Let $u_j$ be a vertex in $X_j$. We can obtain a $k$-sized independent set $Z_j = Y_j \cup \{u_j\}$ from $X_j$ by at most $k-1$ token jumps. 
Next we jump the token on $u_j$ to a vertex in $I_{j+1}\setminus Y_j$ to obtain a $k$-sized independent set $X_{j+1} \subseteq I_{j+1}$. 
This gives us a reconfiguration sequence from $S$ to $T$, as needed. 
\end{proof}

\begin{lemma}\label{lem:reach2}
If there is a reconfiguration sequence $S = I_0, I_1, I_2, \ldots, I_\ell = T$ then there exists a path from $i_s$ to $i_t$ in $\C{G}$.
\end{lemma}

\begin{proof}
Let $I'_1, I'_2, \ldots, I'_{\ell-1}$ be the sets in the covering family such that $I_i \subseteq I'_i$ for $i\in [\ell-1]$. 
Since $|I_i\cap I_{i+1}|=k-1$, we have $|I'_i\cap I'_{i+1}|\geq k-1$. If $I'_i$ and $I'_{i+1}$ are the same set, then they correspond to the same vertex in $\C{G}$. 
Otherwise, they are connected by an edge according to the construction of $\C{G}$. We start from the vertex $i_s$ and follow the reconfiguration sequence until we reach $i_t$. 
This gives us a walk from $i_s$ to $i_t$ and a walk contains a path, as needed.
\end{proof}

The combination of~\Cref{lem:reach1,lem:reach2} give us the following: 

\begin{theorem}
\textsc{Token Jumping Reachability} parameterized by $k$ is fixed-parameter tractable on any graph 
class $\mathscr{C}$ for which we can, given any $n$-vertex graph $G \in \mathscr{C}$, compute a
$k$-independence covering family $\C{F}(G,k)$ of size $\C{O}(f(k) \cdot n^{\C{O}(1)})$ in time $\C{O}(g(k) \cdot n^{\C{O}(1)})$, where $f$ and $g$ are computable functions.
\end{theorem}

\bibliographystyle{plain}
\bibliography{refs}

\begin{thebibliography}{10}

\bibitem{DBLP:reference/algo/AlonYZ16}
Noga Alon, Raphael Yuster, and Uri Zwick.
\newblock Color coding.
\newblock In {\em Encyclopedia of Algorithms}, pages 335--338. 2016.
\newblock \href {https://doi.org/10.1007/978-1-4939-2864-4\_76}
  {\path{doi:10.1007/978-1-4939-2864-4\_76}}.

\bibitem{DBLP:journals/corr/abs-2204-05549}
Valentin Bartier, Nicolas Bousquet, and Amer~E. Mouawad.
\newblock {Galactic Token Sliding}.
\newblock {\em CoRR}, abs/2204.05549, 2022.
\newblock URL: \url{https://arxiv.org/abs/2204.05549}, \href
  {http://arxiv.org/abs/2204.05549} {\path{arXiv:2204.05549}}.

\bibitem{DBLP:journals/mst/BelmonteKLMOS21}
R{\'{e}}my Belmonte, Eun~Jung Kim, Michael Lampis, Valia Mitsou, Yota Otachi,
  and Florian Sikora.
\newblock Token sliding on split graphs.
\newblock {\em Theory Comput. Syst.}, 65(4):662--686, 2021.
\newblock \href {https://doi.org/10.1007/s00224-020-09967-8}
  {\path{doi:10.1007/s00224-020-09967-8}}.

\bibitem{bodlaender_et_al:LIPIcs.IPEC.2021.9}
Hans~L. Bodlaender, Carla Groenland, and C\'{e}line M.~F. Swennenhuis.
\newblock {Parameterized Complexities of Dominating and Independent Set
  Reconfiguration}.
\newblock In Petr~A. Golovach and Meirav Zehavi, editors, {\em 16th
  International Symposium on Parameterized and Exact Computation (IPEC 2021)},
  volume 214 of {\em Leibniz International Proceedings in Informatics
  (LIPIcs)}, pages 9:1--9:16, Dagstuhl, Germany, 2021. Schloss Dagstuhl --
  Leibniz-Zentrum f{\"u}r Informatik.
\newblock URL: \url{https://drops.dagstuhl.de/opus/volltexte/2021/15392}, \href
  {https://doi.org/10.4230/LIPIcs.IPEC.2021.9}
  {\path{doi:10.4230/LIPIcs.IPEC.2021.9}}.

\bibitem{DBLP:conf/wg/BonamyB17}
Marthe Bonamy and Nicolas Bousquet.
\newblock Token sliding on chordal graphs.
\newblock In Hans~L. Bodlaender and Gerhard~J. Woeginger, editors, {\em
  Graph-Theoretic Concepts in Computer Science - 43rd International Workshop,
  {WG} 2017, Eindhoven, The Netherlands, June 21-23, 2017, Revised Selected
  Papers}, volume 10520 of {\em Lecture Notes in Computer Science}, pages
  127--139. Springer, 2017.
\newblock \href {https://doi.org/10.1007/978-3-319-68705-6\_10}
  {\path{doi:10.1007/978-3-319-68705-6\_10}}.

\bibitem{DBLP:conf/swat/BonsmaKW14}
Paul~S. Bonsma, Marcin Kaminski, and Marcin Wrochna.
\newblock Reconfiguring independent sets in claw-free graphs.
\newblock In R.~Ravi and Inge~Li G{\o}rtz, editors, {\em Algorithm Theory -
  {SWAT} 2014 - 14th Scandinavian Symposium and Workshops, Copenhagen, Denmark,
  July 2-4, 2014. Proceedings}, volume 8503 of {\em Lecture Notes in Computer
  Science}, pages 86--97. Springer, 2014.
\newblock \href {https://doi.org/10.1007/978-3-319-08404-6\_8}
  {\path{doi:10.1007/978-3-319-08404-6\_8}}.

\bibitem{DBLP:conf/fct/BousquetMP17}
Nicolas Bousquet, Arnaud Mary, and Aline Parreau.
\newblock Token jumping in minor-closed classes.
\newblock In Ralf Klasing and Marc Zeitoun, editors, {\em Fundamentals of
  Computation Theory - 21st International Symposium, {FCT} 2017, Bordeaux,
  France, September 11-13, 2017, Proceedings}, volume 10472 of {\em Lecture
  Notes in Computer Science}, pages 136--149. Springer, 2017.
\newblock \href {https://doi.org/10.1007/978-3-662-55751-8\_12}
  {\path{doi:10.1007/978-3-662-55751-8\_12}}.

\bibitem{DBLP:journals/corr/abs-2204-10526}
Nicolas Bousquet, Amer~E. Mouawad, Naomi Nishimura, and Sebastian Siebertz.
\newblock A survey on the parameterized complexity of the independent set and
  (connected) dominating set reconfiguration problems.
\newblock {\em CoRR}, abs/2204.10526, 2022.
\newblock \href {http://arxiv.org/abs/2204.10526} {\path{arXiv:2204.10526}},
  \href {https://doi.org/10.48550/arXiv.2204.10526}
  {\path{doi:10.48550/arXiv.2204.10526}}.

\bibitem{DBLP:conf/iwpec/CaiCC06}
Leizhen Cai, Siu~Man Chan, and Siu~On Chan.
\newblock Random separation: {A} new method for solving fixed-cardinality
  optimization problems.
\newblock In Hans~L. Bodlaender and Michael~A. Langston, editors, {\em
  Parameterized and Exact Computation, Second International Workshop, {IWPEC}
  2006, Z{\"{u}}rich, Switzerland, September 13-15, 2006, Proceedings}, volume
  4169 of {\em Lecture Notes in Computer Science}, pages 239--250. Springer,
  2006.
\newblock \href {https://doi.org/10.1007/11847250\_22}
  {\path{doi:10.1007/11847250\_22}}.

\bibitem{DBLP:books/sp/CyganFKLMPPS15}
Marek Cygan, Fedor~V. Fomin, Lukasz Kowalik, Daniel Lokshtanov, D{\'{a}}niel
  Marx, Marcin Pilipczuk, Michal Pilipczuk, and Saket Saurabh.
\newblock {\em Parameterized Algorithms}.
\newblock Springer, 2015.
\newblock \href {https://doi.org/10.1007/978-3-319-21275-3}
  {\path{doi:10.1007/978-3-319-21275-3}}.

\bibitem{DBLP:conf/isaac/DemaineDFHIOOUY14}
Erik~D. Demaine, Martin~L. Demaine, Eli Fox{-}Epstein, Duc~A. Hoang, Takehiro
  Ito, Hirotaka Ono, Yota Otachi, Ryuhei Uehara, and Takeshi Yamada.
\newblock Polynomial-time algorithm for sliding tokens on trees.
\newblock In Hee{-}Kap Ahn and Chan{-}Su Shin, editors, {\em Algorithms and
  Computation - 25th International Symposium, {ISAAC} 2014, Jeonju, Korea,
  December 15-17, 2014, Proceedings}, volume 8889 of {\em Lecture Notes in
  Computer Science}, pages 389--400. Springer, 2014.
\newblock \href {https://doi.org/10.1007/978-3-319-13075-0\_31}
  {\path{doi:10.1007/978-3-319-13075-0\_31}}.

\bibitem{DBLP:journals/siamcomp/DowneyF95}
Rodney~G. Downey and Michael~R. Fellows.
\newblock Fixed-parameter tractability and completeness {I:} basic results.
\newblock {\em {SIAM} J. Comput.}, 24(4):873--921, 1995.
\newblock \href {https://doi.org/10.1137/S0097539792228228}
  {\path{doi:10.1137/S0097539792228228}}.

\bibitem{DBLP:series/mcs/DowneyF99}
Rodney~G. Downey and Michael~R. Fellows.
\newblock {\em Parameterized Complexity}.
\newblock Monographs in Computer Science. Springer, 1999.
\newblock \href {https://doi.org/10.1007/978-1-4612-0515-9}
  {\path{doi:10.1007/978-1-4612-0515-9}}.

\bibitem{DBLP:conf/isaac/Fox-EpsteinHOU15}
Eli Fox{-}Epstein, Duc~A. Hoang, Yota Otachi, and Ryuhei Uehara.
\newblock Sliding token on bipartite permutation graphs.
\newblock In Khaled~M. Elbassioni and Kazuhisa Makino, editors, {\em Algorithms
  and Computation - 26th International Symposium, {ISAAC} 2015, Nagoya, Japan,
  December 9-11, 2015, Proceedings}, volume 9472 of {\em Lecture Notes in
  Computer Science}, pages 237--247. Springer, 2015.
\newblock \href {https://doi.org/10.1007/978-3-662-48971-0\_21}
  {\path{doi:10.1007/978-3-662-48971-0\_21}}.

\bibitem{DBLP:journals/jacm/GroheKS17}
Martin Grohe, Stephan Kreutzer, and Sebastian Siebertz.
\newblock Deciding first-order properties of nowhere dense graphs.
\newblock {\em J. {ACM}}, 64(3):17:1--17:32, 2017.
\newblock \href {https://doi.org/10.1145/3051095} {\path{doi:10.1145/3051095}}.

\bibitem{DBLP:journals/tcs/HearnD05}
Robert~A. Hearn and Erik~D. Demaine.
\newblock {PSPACE}-completeness of sliding-block puzzles and other problems
  through the nondeterministic constraint logic model of computation.
\newblock {\em Theor. Comput. Sci.}, 343(1-2):72--96, 2005.
\newblock \href {https://doi.org/10.1016/j.tcs.2005.05.008}
  {\path{doi:10.1016/j.tcs.2005.05.008}}.

\bibitem{DBLP:conf/ciac/HoangKU19}
Duc~A. Hoang, Amanj Khorramian, and Ryuhei Uehara.
\newblock Shortest reconfiguration sequence for sliding tokens on spiders.
\newblock In Pinar Heggernes, editor, {\em Algorithms and Complexity - 11th
  International Conference, {CIAC} 2019, Rome, Italy, May 27-29, 2019,
  Proceedings}, volume 11485 of {\em Lecture Notes in Computer Science}, pages
  262--273. Springer, 2019.
\newblock \href {https://doi.org/10.1007/978-3-030-17402-6\_22}
  {\path{doi:10.1007/978-3-030-17402-6\_22}}.

\bibitem{DBLP:journals/tcs/ItoDHPSUU11}
Takehiro Ito, Erik~D. Demaine, Nicholas J.~A. Harvey, Christos~H.
  Papadimitriou, Martha Sideri, Ryuhei Uehara, and Yushi Uno.
\newblock On the complexity of reconfiguration problems.
\newblock {\em Theor. Comput. Sci.}, 412(12-14):1054--1065, 2011.
\newblock \href {https://doi.org/10.1016/j.tcs.2010.12.005}
  {\path{doi:10.1016/j.tcs.2010.12.005}}.

\bibitem{DBLP:conf/tamc/ItoKOSUY14}
Takehiro Ito, Marcin Kaminski, Hirotaka Ono, Akira Suzuki, Ryuhei Uehara, and
  Katsuhisa Yamanaka.
\newblock On the parameterized complexity for token jumping on graphs.
\newblock In T.~V. Gopal, Manindra Agrawal, Angsheng Li, and S.~Barry Cooper,
  editors, {\em Theory and Applications of Models of Computation - 11th Annual
  Conference, {TAMC} 2014, Chennai, India, April 11-13, 2014. Proceedings},
  volume 8402 of {\em Lecture Notes in Computer Science}, pages 341--351.
  Springer, 2014.
\newblock \href {https://doi.org/10.1007/978-3-319-06089-7\_24}
  {\path{doi:10.1007/978-3-319-06089-7\_24}}.

\bibitem{DBLP:journals/dam/ItoKOSUY20}
Takehiro Ito, Marcin~Jakub Kaminski, Hirotaka Ono, Akira Suzuki, Ryuhei Uehara,
  and Katsuhisa Yamanaka.
\newblock Parameterized complexity of independent set reconfiguration problems.
\newblock {\em Discret. Appl. Math.}, 283:336--345, 2020.
\newblock \href {https://doi.org/10.1016/j.dam.2020.01.022}
  {\path{doi:10.1016/j.dam.2020.01.022}}.

\bibitem{DBLP:journals/tcs/ItoO19}
Takehiro Ito and Yota Otachi.
\newblock Reconfiguration of colorable sets in classes of perfect graphs.
\newblock {\em Theor. Comput. Sci.}, 772:111--122, 2019.
\newblock \href {https://doi.org/10.1016/j.tcs.2018.11.024}
  {\path{doi:10.1016/j.tcs.2018.11.024}}.

\bibitem{DBLP:journals/tcs/KaminskiMM12}
Marcin Kaminski, Paul Medvedev, and Martin Milanic.
\newblock Complexity of independent set reconfigurability problems.
\newblock {\em Theor. Comput. Sci.}, 439:9--15, 2012.
\newblock \href {https://doi.org/10.1016/j.tcs.2012.03.004}
  {\path{doi:10.1016/j.tcs.2012.03.004}}.

\bibitem{DBLP:conf/coco/Karp72}
Richard~M. Karp.
\newblock Reducibility among combinatorial problems.
\newblock In Raymond~E. Miller and James~W. Thatcher, editors, {\em Proceedings
  of a symposium on the Complexity of Computer Computations, held March 20-22,
  1972, at the {IBM} Thomas J. Watson Research Center, Yorktown Heights, New
  York, {USA}}, The {IBM} Research Symposia Series, pages 85--103. Plenum
  Press, New York, 1972.
\newblock \href {https://doi.org/10.1007/978-1-4684-2001-2\_9}
  {\path{doi:10.1007/978-1-4684-2001-2\_9}}.

\bibitem{DBLP:conf/soda/KreutzerRS17}
Stephan Kreutzer, Roman Rabinovich, and Sebastian Siebertz.
\newblock Polynomial kernels and wideness properties of nowhere dense graph
  classes.
\newblock In Philip~N. Klein, editor, {\em Proceedings of the Twenty-Eighth
  Annual {ACM-SIAM} Symposium on Discrete Algorithms, {SODA} 2017, Barcelona,
  Spain, Hotel Porta Fira, January 16-19}, pages 1533--1545. {SIAM}, 2017.
\newblock \href {https://doi.org/10.1137/1.9781611974782.100}
  {\path{doi:10.1137/1.9781611974782.100}}.

\bibitem{DBLP:journals/talg/LokshtanovM19}
Daniel Lokshtanov and Amer~E. Mouawad.
\newblock The complexity of independent set reconfiguration on bipartite
  graphs.
\newblock {\em {ACM} Trans. Algorithms}, 15(1):7:1--7:19, 2019.
\newblock \href {https://doi.org/10.1145/3280825} {\path{doi:10.1145/3280825}}.

\bibitem{DBLP:journals/jcss/LokshtanovMPRS18}
Daniel Lokshtanov, Amer~E. Mouawad, Fahad Panolan, M.~S. Ramanujan, and Saket
  Saurabh.
\newblock Reconfiguration on sparse graphs.
\newblock {\em J. Comput. Syst. Sci.}, 95:122--131, 2018.
\newblock \href {https://doi.org/10.1016/j.jcss.2018.02.004}
  {\path{doi:10.1016/j.jcss.2018.02.004}}.

\bibitem{DBLP:journals/talg/LokshtanovPSSZ20}
Daniel Lokshtanov, Fahad Panolan, Saket Saurabh, Roohani Sharma, and Meirav
  Zehavi.
\newblock Covering small independent sets and separators with applications to
  parameterized algorithms.
\newblock {\em {ACM} Trans. Algorithms}, 16(3):32:1--32:31, 2020.
\newblock \href {https://doi.org/10.1145/3379698} {\path{doi:10.1145/3379698}}.

\bibitem{DBLP:journals/algorithms/MouawadNRS18}
Amer~E. Mouawad, Naomi Nishimura, Venkatesh Raman, and Sebastian Siebertz.
\newblock Vertex cover reconfiguration and beyond.
\newblock {\em Algorithms}, 11(2):20, 2018.
\newblock \href {https://doi.org/10.3390/a11020020}
  {\path{doi:10.3390/a11020020}}.

\bibitem{DBLP:journals/ejc/NesetrilM08}
Jaroslav Nesetril and Patrice~Ossona de~Mendez.
\newblock Grad and classes with bounded expansion i. decompositions.
\newblock {\em Eur. J. Comb.}, 29(3):760--776, 2008.
\newblock \href {https://doi.org/10.1016/j.ejc.2006.07.013}
  {\path{doi:10.1016/j.ejc.2006.07.013}}.

\bibitem{DBLP:journals/ejc/NesetrilM11a}
Jaroslav Nesetril and Patrice~Ossona de~Mendez.
\newblock On nowhere dense graphs.
\newblock {\em Eur. J. Comb.}, 32(4):600--617, 2011.
\newblock \href {https://doi.org/10.1016/j.ejc.2011.01.006}
  {\path{doi:10.1016/j.ejc.2011.01.006}}.

\bibitem{DBLP:journals/algorithms/Nishimura18}
Naomi Nishimura.
\newblock Introduction to reconfiguration.
\newblock {\em Algorithms}, 11(4):52, 2018.
\newblock \href {https://doi.org/10.3390/a11040052}
  {\path{doi:10.3390/a11040052}}.

\bibitem{DBLP:journals/talg/PhilipRS12}
Geevarghese Philip, Venkatesh Raman, and Somnath Sikdar.
\newblock Polynomial kernels for dominating set in graphs of bounded degeneracy
  and beyond.
\newblock {\em {ACM} Trans. Algorithms}, 9(1):11:1--11:23, 2012.
\newblock \href {https://doi.org/10.1145/2390176.2390187}
  {\path{doi:10.1145/2390176.2390187}}.

\bibitem{DBLP:conf/esa/TelleV12}
Jan~Arne Telle and Yngve Villanger.
\newblock {FPT} algorithms for domination in biclique-free graphs.
\newblock In Leah Epstein and Paolo Ferragina, editors, {\em Algorithms - {ESA}
  2012 - 20th Annual European Symposium, Ljubljana, Slovenia, September 10-12,
  2012. Proceedings}, volume 7501 of {\em Lecture Notes in Computer Science},
  pages 802--812. Springer, 2012.
\newblock \href {https://doi.org/10.1007/978-3-642-33090-2\_69}
  {\path{doi:10.1007/978-3-642-33090-2\_69}}.

\bibitem{DBLP:books/cu/p/Heuvel13}
Jan van~den Heuvel.
\newblock The complexity of change.
\newblock In Simon~R. Blackburn, Stefanie Gerke, and Mark Wildon, editors, {\em
  Surveys in Combinatorics 2013}, volume 409 of {\em London Mathematical
  Society Lecture Note Series}, pages 127--160. Cambridge University Press,
  2013.
\newblock \href {https://doi.org/10.1017/CBO9781139506748.005}
  {\path{doi:10.1017/CBO9781139506748.005}}.

\bibitem{DBLP:conf/iwpec/Zanden15}
Tom~C. van~der Zanden.
\newblock Parameterized complexity of graph constraint logic.
\newblock In Thore Husfeldt and Iyad~A. Kanj, editors, {\em 10th International
  Symposium on Parameterized and Exact Computation, {IPEC} 2015, September
  16-18, 2015, Patras, Greece}, volume~43 of {\em LIPIcs}, pages 282--293.
  Schloss Dagstuhl - Leibniz-Zentrum f{\"{u}}r Informatik, 2015.
\newblock \href {https://doi.org/10.4230/LIPIcs.IPEC.2015.282}
  {\path{doi:10.4230/LIPIcs.IPEC.2015.282}}.

\bibitem{DBLP:journals/jcss/Wrochna18}
Marcin Wrochna.
\newblock Reconfiguration in bounded bandwidth and tree-depth.
\newblock {\em J. Comput. Syst. Sci.}, 93:1--10, 2018.
\newblock \href {https://doi.org/10.1016/j.jcss.2017.11.003}
  {\path{doi:10.1016/j.jcss.2017.11.003}}.

\bibitem{DBLP:journals/tcs/YamadaU21}
Takeshi Yamada and Ryuhei Uehara.
\newblock Shortest reconfiguration of sliding tokens on subclasses of interval
  graphs.
\newblock {\em Theor. Comput. Sci.}, 863:53--68, 2021.
\newblock \href {https://doi.org/10.1016/j.tcs.2021.02.019}
  {\path{doi:10.1016/j.tcs.2021.02.019}}.

\end{thebibliography}

\end{document}